\documentclass[12pt]{article}
\usepackage{amsfonts,amsmath,amssymb}
\usepackage{mathrsfs,mathtools,url}
\usepackage[T1]{fontenc}
\usepackage{graphicx,tikz,color,url}

\newtheorem{theorem}{Theorem}[section]
\newtheorem{definition}[theorem]{Definition}
\newtheorem{lemma}[theorem]{Lemma}
\newtheorem{proposition}[theorem]{Proposition}

\newtheorem{problem}[theorem]{Problem}

\def\cp{\,\square\,}

\DeclareMathOperator{\vv}{vv}

\DeclareMathOperator{\vp}{vp}
\DeclareMathOperator{\str}{str}
\DeclareMathOperator{\MD}{MD}
\DeclareMathOperator{\ecc}{ecc}

\newcommand{\qed}{\hfill $\square$ \bigskip}

\title{The vertex visibility number of graphs}

\author{Dhanya Roy $^{a,}$\footnote{\tt dhanyaroyku@gmail.com, dhanyaroyku@cusat.ac.in} 
   \and Gabriele Di Stefano $^{b,}$\footnote{\tt gabriele.distefano@univaq.it}
	\and Sandi Klav\v{z}ar $^{c,d,e,}$\footnote{\tt sandi.klavzar@fmf.uni-lj.si} 
	\and Aparna Lakshmanan S $^{a,}$\footnote{\tt aparnaren@gmail.com, aparnals@cusat.ac.in} \\\\
	$^{a}$ \small Department of Mathematics, Cochin University of Science and Technology, \\ 
    \small Cochin - 22, Kerala, India \\
    $^b$ \small Department of Information Engineering, Computer Science and Mathematics, \\
    \small University of L'Aquila, Italy \\
	$^{c}$ \small Faculty of Mathematics and Physics, University of Ljubljana, Slovenia\\
	$^{d}$ \small Institute of Mathematics, Physics and Mechanics, Ljubljana, Slovenia \\
	$^{e}$ \small Faculty of Natural Sciences and Mathematics, University of Maribor, Slovenia
}
\date{\today}

\begin{document}
\maketitle
    
\begin{abstract}
If $x\in V(G)$, then $S\subseteq V(G)\setminus\{x\}$ is an $x$-visibility set if for any $y\in S$ there exists a shortest $x,y$-path avoiding $S$. The $x$-visibility number $v_x(G)$ is the maximum cardinality of an $x$-visibility set, and the maximum value of $v_x(G)$ among all vertices $x$ of $G$ is the vertex visibility number ${\rm vv}(G)$ of $G$. It is proved that ${\rm vv}(G)$ is equal to the largest possible number of leaves of a shortest-path tree of $G$. Deciding  whether $v_x(G) \ge k$ holds for given $G$, a vertex $x\in V(G)$, and a positive integer $k$ is NP-complete even for graphs of diameter $2$. Several general sharp lower and upper bounds on the vertex visibility number are proved. The vertex visibility number of Cartesian products is also bounded from below and above, and the exact value of the vertex visibility number is determined for square grids, square prisms, and square toruses. 
\end{abstract}
	
	
\medskip\noindent
{\bf AMS Subj.\ Class.\ (2020)}: 05C12, 68Q17, 05C76

\section{Introduction}

General position and mutual-visibility are very active areas of metric and algorithmic graph theory. The concepts are complementary to each other, and  a progress in one of the areas typically has an impact on the other area. The general position problem has been recently surveyed in~\cite{survey}, see also~\cite{Araujo-2025, KorzeVesel-2025, Roy-2025, Thomas-2024a}. For the mutual-visibility problem we refer to the seminal paper~\cite{distefano-2022} and the following selection of studies~\cite{bresar-2024, bujtas-2025, cicerone-2023a, cicerone-2024b, Klavzar-2025a, Klavzar-2025b, Korze-2024, Kuziak-2023, Roy-2025, Tian-2024a}. 

In~\cite{thank-2024}, the general position version was investigated from a vertex's point of view. Among the many motivations for this research, let us highlight the following. In~\cite[Chapter III]{hardy-2008} it is shown how to place a set of points with integer positive coordinates $(i, j)$, $j \le i$, such that each point is in mutual-visibility with $(0,0)$ and such that the number of points with the same abscissa is maximized. The problem turns out to be very interesting, as the solution has links with the Farey series and the Euler’s totient function $\phi$: the number of points with abscissa $n$ is exactly $\phi(n)$. 

Due to the above motivation (and more), the paper~\cite{thank-2024} gives the following definitions. If $x$ is a vertex of a graph $G$, then $S\subseteq V(G)$ is said to be an {\em $x$-position set} if for any $y \in S$, no vertex of $S \setminus  \{y\}$ lies on any shortest $x,y$-path. The {\em vertex position number} $\vp(G)$ of $G$ is the maximum cardinality of an $x$-position set among all
vertices $x$ of $G$. The paper~\cite{thank-2024} yields numerous results dealing with the largest and smallest orders of maximum $x$-position sets, in particular giving bounds in terms of the girth, vertex degrees, diameter and radius. 

In this paper, we complement the investigation from~\cite{thank-2024} by considering the mutual-visibility problem from a vertex's point of view. If $x$ is a vertex of a graph $G$, then $S\subseteq V(G)\setminus\{x\}$ is an \emph{$x$-visibility set} if for any $y\in S$ there exists a shortest $x,y$-path $P$, such that $V(P)\cap S = \{y\}$. The \emph{$x$-visibility number} $v_x(G)$ is the maximum cardinality of an \emph{$x$-visibility set}, we also say that $v_x(G)$ is the {\em visibility number of $x$}. An \emph{$x$-visibility set} of cardinality \emph{$v_x(G)$} is called a \emph{$v_x$-set}. The maximum value of $v_x(G)$ among all vertices $x$ of $G$ is called the \emph{vertex visibility number} $\vv(G)$ of $G$, and a corresponding $v_x$-set is a {\em $\vv$-set}. For instance, if $n\ge 3$, then $v_x(C_n) = 2$ for any $x \in V(C_n)$, hence $\vv(C_n) = 2$. Similarly, if $n\ge 3$, then $v_x(P_n) = \deg_{P_n}(x)$, hence $\vv(P_n) = 2$. It is worth mentioning  the following fact. 

\begin{lemma}\label{non-leaf}
If $x$ is a leaf of a connected graph $G$ with $n(G) \geq 3$, then $v_x(G) < \vv(G)$.
\end{lemma}

\begin{proof}
If $G$ is a path graph on at least three vertices, then $v_x(G) =1$ and $ \vv(G) = 2$. Assume next that $G$ is not a path, let $y$ be the support vertex of $x$, and $S$ be a $v_x$-set of $G$. Then $S' = S\cup \{x\}$ is a $y$-visibility set, hence $\vv(G) \ge |S'| > |S| = v_x(G)$. 
\qed
\end{proof}

The paper is organized as follows. In the following subsection, further definitions needed are listed. In Section~\ref{sec:complexity} we first prove the fundamental reason for the computational difficulty of the vertex visibility number:  $\vv(G)$ is equal to the largest possible number of leaves of a shortest-path tree of $G$. 
This builds a bridge between the new introduced concepts and the Dijkstra's algorithm, the result of which is precisely a tree of shortest-paths.
In the second main result of the section we prove that the {\sc $x$-visibility} problem is NP-complete even for graphs of diameter $2$, where the {\sc $x$-visibility} problem asks whether $v_x(G) \ge k$ holds for given $G$, a vertex $x\in V(G)$, and a positive integer $k$. In Section~\ref{sec:general} we deduce some general sharp lower and upper bounds on the vertex visibility number in terms of the order, the maximum degree, and the eccentricity. The vertex visibility number of Cartesian products is investigated in Section~\ref{sec:Cart}. After establishing a general sharp lower bound and an upper bound, the focus is on square grids, square prisms, and square toruses. For each of these Cartesian products the exact value of the vertex visibility number is obtained. We conclude the paper with some problems for future investigation. 

\subsection{Concepts and notation}

Here we provide further necessary definitions. 

Let $G=(V(G), E(G))$ be a graph. The order of $G$ is denoted by $n(G)$, its maximum degree by $\Delta(G)$, and the (open) neighborhood of $x\in V(G)$ by $N_G(x)$. A vertex $x$ of $G$ is {\em simplicial} if the subgraph induced by $N_G(x)$ is complete. A vertex of $G$ is {\em universal} if it is adjacent to all other vertices of $G$. A {\em double star} is a tree in which exactly two vertices are not leaves. 

The distance function $d_G(\cdot, \cdot)$ is the standard shortest-path distance. The {\em eccentricity} $\ecc_G(x)$ of a vertex $x$ is the maximum distance between $x$ and the other vertices of $G$. A vertex $z$ is an {\em eccentric vertex} of $x$ if $d_G(x,z) = \ecc_G(x)$. The open interval $I(x,y)$ between vertices $x$ and $y$ is the set of all vertices that lie on shortest $x,y$-paths other than $x$ and $y$. When every two vertices $x$ and $y$ of $G$ are connected by a unique shortest $x,y$-path, $G$ is called {\em geodetic}. A vertex $y$ is {\em maximally distant} from $x$ if $d_G(x,y) \geq d_G(x,z)$, for every $z \in N_G(y)$. (This notion was introduced in~\cite{Oellermann-2007} as a tool to study the strong metric dimension.) The collection of all maximally distant vertices from $x$ is denoted by $\MD_G(x)$. 

If $S\subseteq V(G)$, then we say that $x,y\in V(G)$ are {\em $S$-visible}, if there exists a shortest $x,y$-path $P$ such that $V(P)\cap S \subseteq \{x,y\}$. The set $S$ is a \emph{mutual-visibility set} if any two vertices from $S$ are $S$-visible. The cardinality of a largest mutual-visibility set of $G$ is the {\em mutual-visibility number} $\mu(G)$ of $G$. A mutual-visibility set of cardinality $\mu(G)$ is a {\em $\mu$-set}. We also say that a vertex $y$ is {\em $S$-visible from} $x$ if there exists a shortest $y,x$-path $P$ such that $V(P)\cap S \subseteq \{x, y\}$. 

The \emph{Cartesian product} $G\cp H$ of graphs $G$ and $H$ has the vertex set $V(G)\times V(H)$, vertices $(g,h)$ and $(g',h')$ of $G\cp H$ are adjacent if either $gg'\in E(G)$ and $h=h'$, or $g=g'$ and $hh'\in E(H)$. A {\em $G$-layer} is a subgraph of $G\cp H$ induced by $V(G)\times \{h\}$ for some $h\in V(H)$, denoted by $G^h$. Analogously, for $g\in V(G)$ we have the {\em $H$-layer} $^{g}H$.

Finally, for a positive integer $k$, the set $\{1,\dots ,k\}$ is denoted by $[k]$. 

\section{Shortest-path trees and computational complexity}
\label{sec:complexity}

To study the computational complexity of finding the $x$-visibility number $v_x(G)$ for a fixed vertex $x$ of a graph $G$, we introduce the following decision problem.

\begin{definition}
{\sc $x$-visibility} problem: \\
{\sc Instance}: A graph $G$, a vertex $x\in V(G)$, a positive integer $k\leq n(G)$. \\
{\sc Question}: $v_x(G)\geq k$?
\end{definition}

In this section we prove that the {\sc $x$-visibility} problem is NP-complete even for graphs of diameter $2$, so that finding $\vv(G)$ is an NP-hard problem. This result is in sharp contrast with the computational complexity of the vertex position problem, which is solvable in $O(n^4\log (n))$ time, as shown in~\cite{thank-2024}. 

The intrinsic difficulty of the {\sc $x$-visibility} problem lies in the fact that the value $\vv(G)$ is equal to the largest possible number of leaves of a shortest-path tree of $G$. These trees are defined as follows. Let $G$ be a graph and $x\in V(G)$. Then a tree rooted in $x$ constructed by the BFS search is called a {\em shortest-path tree} of $G$. That is, it is a rooted spanning tree $T$, such that $d_T(x,y) = d_G(x,y)$ for every $y\in V(G)$. 

\begin{theorem}	
\label{thm:equivalnet-problem}
If $G$ is a connected graph, then $\vv(G)$ is equal to the largest possible number of leaves of a shortest-path tree of $G$.
\end{theorem}

\begin{proof}
Let $S = \{x_1, \dots, x_k\}$ be a $\vv$-set of $G$, and let $x$ be a vertex with $v_x(G) = \vv(G)$. Then, if $n(G)\ge 3$, $x$ is not a leaf because $v_{x'}(G) > v_x(G)$, where $x'$ is the (support) vertex adjacent to $x$.  

Denoting the largest possible number of leaves of a shortest-path tree of $G$ by $t_G$, it is clear that $\vv(G) \ge t_G$. 

To prove that $t_G \ge \vv(G)$, we first claim that there exist shortest $x,x_i$-paths $P_i$, $i\in [k]$, such that $V(P_i) \cap S = \{x_i\}$ and the edges from $\cup_{i\in [k]}E(P_i)$ induce a tree $T$ in $G$. Since $S$ is a $v_x$-set, there exist shortest $x,x_i$-paths $P_i$, $i\in [k]$, such that $V(P_i) \cap S = \{x_i\}$. Assume that the paths $P_i$ are selected such that the union $\cup_{i\in [k]}E(P_i)$ induces the smallest number of cycles possible. Suppose that this number is not zero. Then there exist $s,t\in [k]$ such that the union of edges $E(P_s) \cup E(P_{t})$ induces a cycle. Let $y$ be a vertex from $V(P_s)\cap V(P_{t})$ such that $d_G(x,y)$ is as large as possible among all vertices from $V(P_s)\cap V(P_t)$. Let $P_s'$ and $P_t'$ be the subpaths of $P_s$ and $P_t$ between $x$ and $y$, respectively. Since $P_s$ and $P_t$ are shortest paths,  $P_s'$ and $P_t'$ are shortest paths as well. Then in our collection $\{P_i:\ i\in [k]\}$ of shortest paths, replace $P_t$ with the concatenation of $P_s'$ and the $y,x_t$-subpath of $P_t$. In this way the union of the edges from the new collections of shortest $x,x_i$-paths induces at least one cycle less. This contradicts the selection of the paths $P_i$. We can conclude that the union $\cup_{i\in [k]}E(P_i)$ induces no cycle, thus proving the claim.  

Let $T$ be the tree induced by the edges from $\cup_{i\in [k]}E(P_i)$. If $T$ is a spanning tree, then we are done. Hence assume that $T$ does not span $G$. Let $y$ be an arbitrary vertex from $V(G)\setminus V(T)$. Then every shortest $x,y$-path contains at least one vertex from $S$, for otherwise $S \cup \{y\}$ would be an $x$-position set larger than $S$, which is not possible. Moreover, there exists a shortest $x,y$-path $P$ such that $V(P) \cap S = \{x_j\}$ for some $j\in [k]$, for otherwise $S$ would not be an $x$-position set. Then replace $x_j$ by $y$ in $S$ and in the collection of our shortest paths, replace the current shortest $x,x_j$-path by $P$. Proceeding in this way, we end with an $x$-position set of the same cardinality as $S$, and with a collection of paths whose edges induce a spanning tree $T$ of $G$. We can conclude that $t_G \ge \vv(G)$.
\qed
\end{proof}

In view of Theorem~\ref{thm:equivalnet-problem}, we first recall that the problem of finding a spanning tree with maximum number of leaves has been heavily researched and is computationally difficult. For instance, the problem is NP-hard as well as APX-hard for cubic graphs, see~\cite{bonsma-2012}. On the other hand, a $2$-approximation algorithm is known for this problem~\cite{solis-2017}. We now demonstrate that the difference between the maximum number of leaves in a spanning tree of $G$ and $\vv(G)$ can be arbitrarily large. 
    
\begin{figure}[ht!]
\begin{center}
\begin{tikzpicture}[scale=0.5,style=thick,x=1.5cm,y=1.5cm]
\def\vr{5pt}

\begin{scope}[xshift=0cm, yshift=0cm] 
\coordinate(x1) at (0,0);
\coordinate(x2) at (1,0);
\coordinate(x3) at (2,0);
\coordinate(x4) at (0,1);
\coordinate(x5) at (1,1);
\coordinate(x6) at (2,1);
\coordinate(x7) at (1,2);

\coordinate(x8) at (3,-1);

\coordinate(x9) at (4,0);
\coordinate(x10) at (5,0);
\coordinate(x11) at (6,0);
\coordinate(x12) at (4,1);
\coordinate(x13) at (5,1);
\coordinate(x14) at (6,1);
\coordinate(x15) at (5,2);

\coordinate(x16) at (3,3);

\draw (x5) -- (x1);
\draw (x5) -- (x2);
\draw (x5) -- (x3);
\draw (x7) -- (x4);
\draw (x7) -- (x5);
\draw (x7) -- (x6);

\draw (x13) -- (x9);
\draw (x13) -- (x10);
\draw (x13) -- (x11);
\draw (x15) -- (x14);
\draw (x15) -- (x13);
\draw (x15) -- (x12);

\draw (x16) -- (x7);
\draw (x16) -- (x8);
\draw (x16) -- (x15);

\draw (x8) -- (x2);
\draw (x8) -- (x10);

\draw(x1)[fill=black] circle(\vr);
\draw(x2)[fill=black] circle(\vr)node[left]{{\footnotesize $b$}};
\draw(x3)[fill=black] circle(\vr);
\draw(x4)[fill=black] circle(\vr);
\draw(x5)[fill=white] circle(\vr)node[left]{{\footnotesize $a$}};
\draw(x6)[fill=black] circle(\vr);
\draw(x7)[fill=white] circle(\vr)node[above]{{\footnotesize $y$}};
\draw(x8)[fill=white] circle(\vr)node[below]{{\footnotesize $c$}};
\draw(x9)[fill=black] circle(\vr);
\draw(x10)[fill=black] circle(\vr);
\draw(x11)[fill=black] circle(\vr);
\draw(x12)[fill=black] circle(\vr);
\draw(x13)[fill=white] circle(\vr);
\draw(x14)[fill=black] circle(\vr);
\draw(x15)[fill= white] circle(\vr) node[above]{{\footnotesize $z$}};
\draw(x16)[fill=white] circle(\vr)node[above]{{\footnotesize $x$}};
\end{scope}

\end{tikzpicture}
\caption{A graph $G$ with $\ell(G) = 11$ but $\vv(G) = \mu_x(G) = 10$ (also attained at $y$ and $z$).}
\label{fig:vertex-visibility-counter-eg}
\end{center}
\end{figure}
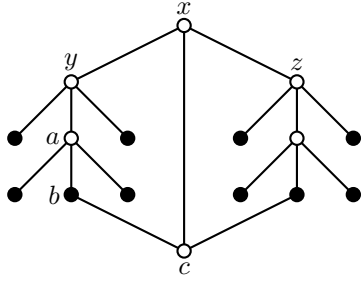

Let $\ell(G)$ denote the maximum number of leaves in a spanning tree of a graph $G$. Let $G_n$ denote the graph obtained by taking $n$ copies of the graph $G$ in Figure~\ref{fig:vertex-visibility-counter-eg} and identifying the vertex $x$. Then $\ell(G_n) = 11n$.  By Lemma \ref{lem:mdv}, it is straight forward to check that $v_x(G_n) = v_y(G_n) = v_z(G_n) = 10n$ while $v_a(G_n) = 9$, $v_b(G_n) = 8$ and $v_c(G_n) = 8$. Thus by the symmetry of the graph and by Lemma \ref{non-leaf}, we can conclude that $\vv(G_n) = 10n$. Hence the difference $\ell(G_n) - \vv(G_n)$ is $n$, which we summarize in the following result.

\begin{proposition}
\label{prop:big-diff}
Given a graph $G$, the difference between $\ell(G)$ and $vv(G)$ can be arbitrarily large.
\end{proposition}

We are now ready to prove that the {\sc $x$-visibility} problem is hard to solve. 

\begin{theorem}\label{thm:NP}
The {\sc $x$-visibility} problem is NP-complete even for graphs of diameter $2$.
\end{theorem}

\begin{proof}
Given a set $X\subseteq V(G)\setminus \{x\}$ of $G$, it is possible to test in polynomial time whether it is a $x$-visibility set or not. Consequently, the problem is in NP. 

We will now prove that the {\sc Independent Set} problem (equivalent to the {\sc Clique} problem in the complement graph and shown as NP-complete in~\cite{Karp72}), polynomially reduces to the {\sc $x$-visibility} problem. The {\sc Independent Set} problem asks if the independence number $\alpha(G)$ of a graph $G$, that is the cardinality of a largest edgeless set of vertices of $G$, is at least a given integer $t$. Here we assume that $G$ does not contain any isolated vertices. Clearly, this assumption does not affect the hardness of the problem.
We consider an arbitrary instance $(G,t)$ of the {\sc Independent Set} problem where $G$ has no isolated vertices and will construct an instance$(G',k)$ of the {\sc $x$-visibility} problem as follows.

Let $V(G) = [n]$ and add a universal vertex $x$ to $V(G)$. Now for each edge $e=ij$ of $G$, we add a new vertex $v_e=v_{ij}$ and the edges $iv_e$ and $jv_e$. Also, let $S = \{v_e: e \in E(G)\}$ induce a clique. An example of the graph $G'$, for $G = P_5$, is given in Figure~\ref{fig:NP-completeness}.


\begin{figure}[ht!]
\begin{center}
\begin{tikzpicture}[scale=0.5,style=thick,x=6cm,y=2cm]
\def\vr{5pt}

\begin{scope}[xshift=0cm, yshift=0cm] 
\coordinate(x1) at (1.5,0);
\coordinate(x2) at (1.5,1);
\coordinate(x3) at (2.5,1);
\coordinate(x4) at (2.5,0);

\coordinate(x5) at (0,2);
\coordinate(x6) at (1,2);
\coordinate(x7) at (2,2);
\coordinate(x8) at (3,2);
\coordinate(x9) at (4,2);

\coordinate(x10) at (2,3);

\draw (x1) -- (x2) -- (x3) -- (x4) -- (x1);
\draw (x2) -- (x4);
\draw (x1) -- (x3);
\draw (x5) -- (x6) -- (x7) -- (x8) -- (x9);
\draw (x5) -- (x1) -- (x6);
\draw (x6) -- (x2) -- (x7);
\draw (x7) -- (x3) -- (x8);
\draw (x8) -- (x4) -- (x9);
\draw (x10) -- (x5);
\draw (x10) -- (x6);
\draw (x10) -- (x7);
\draw (x10) -- (x8);
\draw (x10) -- (x9);

\draw(x1)[fill=black] circle(\vr)node[below=1mm]{\footnotesize $v_{01}$};
\draw(x2)[fill=black] circle(\vr)node[left]{\footnotesize$v_{12}$};
\draw(x3)[fill=black] circle(\vr)node[right]{\footnotesize$v_{23}$};
\draw(x4)[fill=black] circle(\vr)node[below=1mm]{\footnotesize$v_{34}$};
\draw(x5)[fill=black] circle(\vr)node[below=1mm]{\footnotesize$0$};
\draw(x6)[fill=white] circle(\vr)node[below=1mm]{\footnotesize$1$};
\draw(x7)[fill=black] circle(\vr)node[below=1mm]{\footnotesize$2$};
\draw(x8)[fill=white] circle(\vr)node[below=1mm]{\footnotesize$3$};
\draw(x9)[fill=black] circle(\vr)node[below=1mm]{\footnotesize$4$};
\draw(x10)[fill=white] circle(\vr)node[above=1mm]{\footnotesize$x$};
\end{scope}

\end{tikzpicture}
\caption{The graph $G'$ for $G=P_5$.}
\label{fig:NP-completeness}
\end{center}
\end{figure}
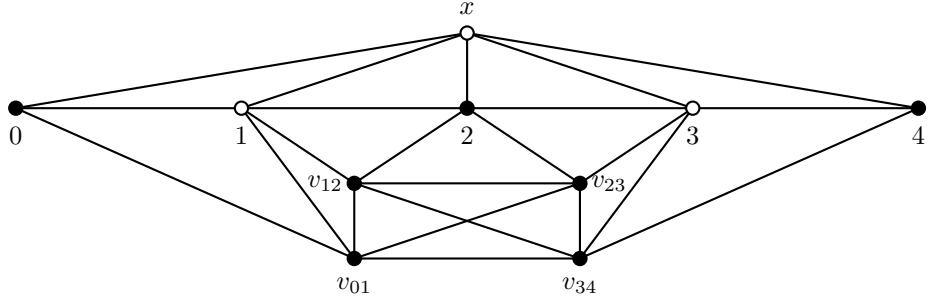

Let $(G,t)$ be an instance of the {\sc Independent Set} problem, we show that $\alpha(G)\geq t$ if and only if the {\sc $x$-visibility} problem has $x$-visibility number greater than $k$, where $k= E(G)+t$, that is $v_x(G')\geq k$.

Let $I$ be an independent set of $G$, and consider the set $X=I\cup S$ in $G'$. For each $v_e \in S$, where $e=ij$, either $i$ or $j$ is not present in $I$ and hence all vertices in $S$ are $x$-visible. Also, all vertices in $I$ are adjacent to $x$ and hence are $x$-visible. Therefore, $v_x(G') \geq k$.

On the contrary, let $X$ be a $x$-visibility set of $G'$ of cardinality at least $k$. By Lemma \ref{lem:mdv}, we can assume that all vertices of $S$ are in $X$. Then the vertices in $X\setminus S$ must form an independent set of $G$, as, otherwise, there would be an edge $e=ij$ such that both vertices $i$ and $j$ are in $X$, but then $v_e$ would not be in visibility with $x$. Since the number of vertices in $S$ is $|E(G)|$, the number of vertices in the independent set of $G$ is $|X\setminus S|\geq k - |E(G)| = t$. Hence, $\alpha(G)\geq t$.

Finally, the diameter of $G'$ is two as $x$ is adjacent to all vertices of $G$ and at distance 2 from the vertices in $S$. All pairs of vertices $u,v$ in $G$ are at distance at most two as $u,x,v$ is a path in $G'$. Each vertex in $S$ is adjacent to all other vertices in $S$ and, since $G$ has no isolated vertices, it is at distance at most two from the vertices in $G$. 
\qed
\end{proof}

The above Theorem~\ref{thm:NP} clearly implies that also finding $\vv(G)$ is an NP-hard problem.

\section{General lower and upper bounds}
\label{sec:general}

In view of the hardness of the problem studied, as established in Section~\ref{sec:complexity}, in this section we prove general lower and upper bounds for the vertex visibility numbers. Along the way, several exact values are also determined. 

If $S$ be a $\mu$-set of $G$ and $x \in S$, then $S\setminus \{ x\} $ is an $x$-visibility set of $G$. Also, if $x\in V(G)$, then $N_G(x)$ is an $x$-visibility set. Thus we have the following general bounds:  
\begin{equation}
\label{eq:general}
\max\{\mu(G) - 1, \Delta(G)\} \leq \vv(G) \leq n(G)-1.  
\end{equation}

The difference $\vv(G) - (\mu(G) -1)$ can be arbitrarily large. For instance, it can be deduced from~\cite[Theorem 3.2.(i)]{cicerone-2023} that if $n \geq 6$, then $\mu(K_2 \cp C_n) = 6$, while on the other hand $\vv(K_2 \cp  C_n) \geq n$. The first assertion of the following result follows from~\eqref{eq:general}, while the second is also straightforward to verify, hence we omit the proof. 

\begin{proposition}
\label{prop:n-1 and n-2}
If $G$ is a connected graph with $n(G)\ge 2$, then the following properties hold. 
\begin{enumerate}
\item $\vv(G) = n(G) - 1$ if and only if $G$ has a universal vertex.
\item $\vv(G) = n(G) - 2$ if and only if $G$ contains no universal vertex and contains a spanning double star.
\end{enumerate}
\end{proposition}

We now say that a vertex $y$ is a \emph{stress vertex} for $x$ if there exists a maximally distant vertex $z$ of $x$ such that $y$ lies on every $x,z$-shortest path in $G$. Note that, every cut vertex $y$, $y\ne x$, is a stress vertex for $x$. Let $\str_G(x)$ denote the number of stress vertices for $x$. 

\begin{lemma}\label{lem:mdv}
If $x$ is a vertex of a connected graph $G$, then there exists a $v_x$-set $S$ with the following properties:
    \begin{enumerate}
        \item every $y \in V(G)\setminus \{x\}$ is $S$-visible from $x$; 
        \item $\MD_G(x) \subseteq S$; 
        \item $S$ contains no stress vertex for $x$.
    \end{enumerate}
In addition, $|\MD_G(x)| \leq v_x(G) \leq n(G) - \str_G(x) - 1$. 
\end{lemma}

\begin{proof}
Let $R$ be a $v_x$-set of $G$ such that it contains as many vertices from $\MD_G(x)$ as possible. Suppose there exists a vertex $y \in \MD_G(x) \setminus (R\cup\{x\})$. Then every shortest $x,y$-path intersects $R$. Also, there exists a shortest $x,y$-path that contains exactly one vertex from $R$, say $z$, for otherwise $R$ would not be an $x$-visibility set. But then $(R\setminus \{z\}) \cup \{y\}$ is an an $x$-visibility set of the same cardinality as $R$ containing more vertices from $\MD_G(x)$, a contradiction. 

By the above we have $\MD_G(x) \subseteq R$. If follows from here by the definition of the stress vertex that $R$ contains no stress vertex for $x$ as well as that every $y \in V(G)\setminus \{x\}$ is $R$-visible from $x$. The set $R$ is thus a a $v_x$-set with the three required properties. Finally, the second property yields the claimed lower bound, while the third property gives the  upper bound. 
\qed 
\end{proof}

Note that Lemma~\ref{lem:mdv} implies that if $G$ is a geodetic graph, then $v_x(G) = |\MD_G(x)|$ for every $x\in V(G)$. For the block graphs, which form a special case of geodetic graph, more specific result can be stated. In fact, this is an example where the lower and the upper bound in Lemma~\ref{lem:mdv} coincide. Denoting by $s(G)$ the number of simplicial vertices of a graph $G$, the result for block graphs reads as follows. 

\begin{proposition}
\label{prop:block}
If $G$ is a block graph different from a complete graph then, $$\vv(G) = s(G).$$
\end{proposition}

\begin{proof}
Since $G$ is not complete, it contains a cut vertex $x$. Then $\MD_G(x)$ consists of the simplicial vertices of $G$. By the lower bound of Lemma~\ref{lem:mdv} we thus get $\vv(G)\ge s(G)$. On the other hand, as already mentioned, cut vertices are stress vertices, hence by the upper bound of Lemma~\ref{lem:mdv} we also get $\vv(G)\le s(G)$.
\qed 
\end{proof}
 
We continue with the next upper bound depending on the order and the maximum degree. 

\begin{theorem}
\label{thm:upper-Delta}
If $G$ has no universal vertex, then 
$$\vv(G) \leq \left\lfloor \frac{n(G)\Delta(G) - 1}{\Delta(G) + 1}\right\rfloor,$$
and the bound is sharp.
\end{theorem}

\begin{proof}
Let $S$ be a $\vv$-set of $G$ and set $n = n(G)$ and $\Delta = \Delta(G)$ for the rest of the proof. 

Assume first that $\vv(G) = \Delta$. Then we have 
$$\vv(G) = \Delta \leq \frac{n\Delta-1}{\Delta+1}\,,$$
where the above inequality holds since $\Delta \le n-2$ by the theorem's assumption. 

Assume second that $\vv(G) > \Delta$. Let $x$ be a vertex such $v_x(G) = \vv(G)$ and let $S$ be a $v_x$-set. Then $|S| = \vv(G)$. Let $S' = S\setminus N_G(x)$ and let $s' = |S'|$. Since $\vv(G) > \Delta$ we have $s'\ge 1$. This in turn implies that at least one neighbor of $x$ does not belong to $S$. Consequently, 
\begin{equation}
\label{eq:first}    
s' \ge \vv(G) - (\deg_G(x)-1) \ge \vv(G) - \Delta + 1\,.
\end{equation}
Since $S$ is an $x$-visibility set, each vertex from $S'$ has a neighbor in $V(G)\setminus (S\cup \{x\})$. It follows that 
\begin{equation}
\label{eq:second}    
s' \le (n - \vv(G) - 1) \Delta\,.
\end{equation}
Combining~\eqref{eq:first} with~\eqref{eq:second} we get
$$\vv(G) - \Delta + 1 \leq (n- \vv(G) - 1) \Delta\,.$$ 
Rearranging this inequality and using the fact that $\vv(G)$ is a positive integer, the claimed inequality follows. 
    
To show sharpness of the bound, consider the cocktail party graph $K_{k\times2}$, (the complete $k$-partite graph where each partite set is of cardinality 2) of order $2k$, where $k\ge 2$. Then $\Delta(K_{k\times2}) = 2k-2 = \vv(K_{k\times2}) = 2k-2$. Since 
$$\left\lfloor \frac{2k(2k-2) - 1}{(2k-2) + 1} \right\rfloor = 
\left\lfloor \frac{2k - 2 - \frac{1}{2k}}{1 - \frac{1}{2k}} \right\rfloor = 2k - 2\,,$$
the equality holds for $K_{k\times2}$, $k\ge 2$. 
\qed
\end{proof}

We next bound the $x$-visibility number using the eccentricity of $x$. 

\begin{proposition}
If $G$ is a graph and $x\in V(G)$, then 
$$\frac{n(G)-1}{\ecc_G(x)}\leq v_x(G)\leq n(G) - \ecc_G(x)\,,$$ 
and the bounds are sharp.
\end{proposition}

\begin{proof}
Let $S$ be a $v_x$-set, and let $y$ be an eccentric vertex of $x$ in $G$. Let $P:x=x_0,x_1,\dots,x_k=y$ be a shortest $x,y$-path selected such that $|V(P)\cap S|$ is smallest possible. Suppose that $S$ contains more than one vertex from $P$, say $x_i$ and $x_j$, where $i<j$. Since $S$ is an $x$-visibility set, there exists a shortest $x,x_j$-path $Q$ such that $V(Q)\cap S=\{x_j\}$. This in turn implies that the path $Q$ followed by the $x_j,y$-subpath of $P$ is a shortest $x,y$-path which contains fewer vertices from $S$ than $P$, a contradiction to the choice of $P$. Hence $|V(P)\cap S| \le 1$ which in turn yields $v_x(G)\leq n(G)-\ecc_G(x)$. For the lower bound, let $N_i(x)$, $i\in [\ecc_G(x)]$, be the the set of all vertices at distance exactly $i$ from $x$. Then each of these is an $x$-visibility set, and one of them must have order at least $(n(G)-1)/ \ecc_G(x)$.

If $x$ is the universal vertex in a star graph, then the lower and the upper bound coincide.
\qed
\end{proof}

\section{Cartesian products}
\label{sec:Cart}	

In this section, we consider the Cartesian product graphs and first prove general lower and upper bounds on their vertex visibility number. Next, we focus on square grids ($P_n\cp P_n)$, square prisms ($P_n\cp C_n$), and square toruses ($C_n\cp C_n$), and determine exact values for the vertex visibility number in all three cases. The general bounds read as follows. 

\begin{proposition}\label{prop:Cart}
If $G$ and $H$ are graphs with $n(G) \geq n(H)$, then 
$$\max \{\Delta(G)n(H), \Delta(H)n(G)\} \leq \vv(G \cp H) \leq (n(G)-1)n(H)\,.$$ 
Both bounds are sharp, in particular, $\vv(K_m\cp K_n) =mn-m$ for $m\geq n\geq 2$. 
\end{proposition}
	
\begin{proof}
Let $g \in V(G)$ be such that $\deg_G(g) = \Delta(G)$ and let $h$ be an arbitrary vertex of $H$. Then $N_G(x) \times V(H)$ is a $(g,h)$-visibility set of $G \cp H$ with cardinality $\Delta(G)n(H)$. By the commutativity of the Cartesian product operation, there also exists a $(g',h')$-visibility set of $G \cp H$ with cardinality $\Delta(H)n(G)$. Thus, $\vv(G \cp H) \geq \max\{\Delta(G)n(H), \Delta(H)n(G)\}$.

To prove that $\vv(G \cp H) \leq (n(G)-1)n(H)$ suppose on the contrary that $G\cp H$ contains a $(g,h)$-visibility set $S$ with $|S| > (n(G)-1)n(H)$. Then there exist $g'\in V(G)$ and $h'\in V(H)$ such that $V(G^{h'}) \subseteq S$ and $V(^{g'}H) \subseteq S$. But then the vertex $(g',h') \in S$ is not $S$-visible from $(g,h)$, a contradiction. 

If $G$ has a universal vertex, then the lower and the upper bound coincide. A particular case of this situation is the product $K_n\cp K_m$. 
\qed
\end{proof}

Next, we will cover each of square grids, square prisms, and square toruses in their own subsections. 

\subsection{Square grids}
\label{subsec:grids}	

In this subsection we determine the vertex visibility number of square grids.    
    
\begin{theorem}\label{vv-sqauregrid}
If $n \ge 4$, then 
$$\vv(P_n \cp P_n) = \frac{n^2+n-2}{2}\,.$$
\end{theorem}

\begin{proof}
Let $V(P_n \cp P_n) = \{(u_k,v_l): k,l \in [n]\}$. Consider an arbitrary vertex $x = (u_{i+1},v_{j+1}) \in V(P_n \cp P_n)$, where we may assume without loss of generality that $0 \leq i \leq j \leq \lfloor \frac{n-1}{2} \rfloor$. Let $X = V(^{u_{i+1}}P_n)$  and $Y = V(P_n^{v_{j+1}})$. Set further
\begin{align*}
Q_1 & = \{(u_k,v_l):\ k < i+1, l < j+1\}, \\
Q_2 & = \{(u_k,v_l):\ k < i+1, l > j+1\}, \\
Q_3 & = \{(u_k,v_l):\ k > i+1, l > j+1\}, \\
Q_4 & = \{(u_k,v_l):\ k > i+1, l < j+1\}. 
\end{align*}
Let $S$ be a $v_x$-set in $P_n \cp P_n$ that satisfies the properties given in Lemma \ref{lem:mdv} and has the maximum number of vertices from $\bigcup_{t=1}^{4} Q_t$ among its choices.

For $(u_k,v_l) \in Q_1$, a shortest $(u_{i+1},v_{j+1}),(u_k, v_l)$-path passes through either $(u_k,v_{l+1})$ or $(u_{k+1},v_l)$, and hence these two vertices cannot be simultaneously in $S$. A similar argument is valid for $(u_k,v_l) \in Q_t$, $t \in [4]$. Now, partition each $Q_t$ into diagonals $D_{pq}$, where $p, q \in Q_t$ are its end vertices, as illustrated in Figure~\ref{fig:grid}. The rectangles indicate the sets $Q_i$, $i\in [4]$, and the doted lines indicate the diagonals $D_{pq}$. 

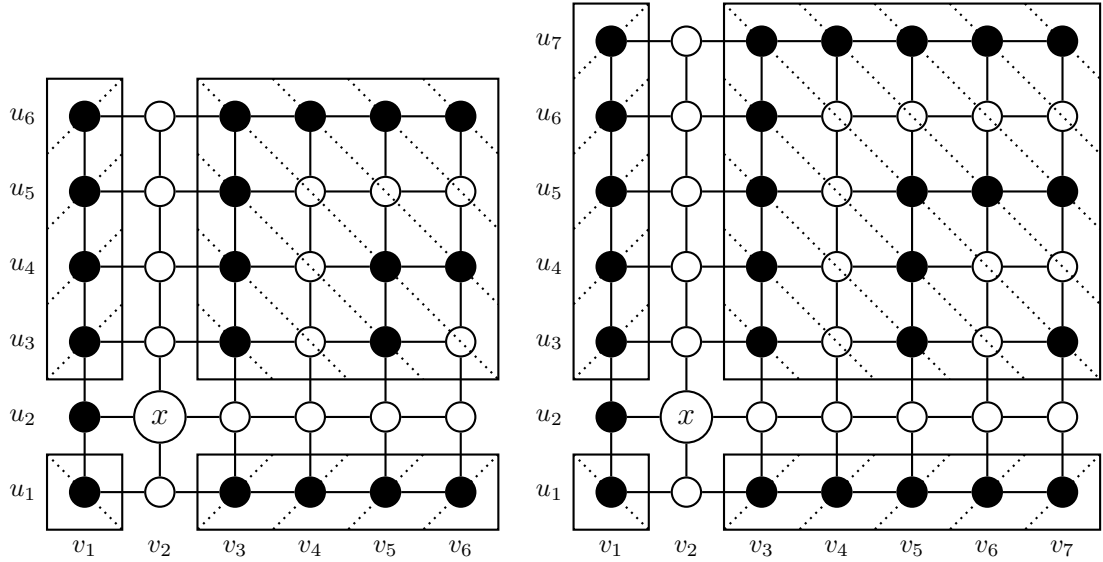
\begin{figure}[ht!]
	\begin{center}
		\begin{tikzpicture}[scale=0.5,style=thick,x=2cm,y=2cm]
			\node[circle, draw, fill = black](a) at (0,0) {};
            \node at (a.west) [left=8pt] {\footnotesize $u_1$};
            \node at (a.south) [below=8pt] {\footnotesize $v_1$};
			\node[circle, draw] (b) at (1,0) {};
            \node at (b.south) [below=8pt] {\footnotesize $v_2$};
			\node[circle, draw, fill = black] (c) at (2,0) {};
            \node at (c.south) [below=8pt] {\footnotesize $v_3$};
			\node[circle, draw, fill = black] (d) at (3,0) {}; 
            \node at (d.south) [below=8pt] {\footnotesize $v_4$};
			\node[circle, draw, fill = black] (e) at (4,0) {}; 
            \node at (e.south) [below=8pt] {\footnotesize $v_5$};
			\node[circle, draw, fill = black] (f) at (5,0) {}; 
            \node at (f.south) [below=8pt] {\footnotesize $v_6$};
			\draw (a) -- (b) -- (c) -- (d) -- (e) -- (f);
			\node[circle, draw, fill = black] (a1) at (0,1) {};
            \node at (a1.west) [left=8pt] {\footnotesize $u_2$};
			\node[circle, draw] (b1) at (1,1) {$x$};
			\node[circle, draw] (c1) at (2,1) {};
			\node[circle, draw] (d1) at (3,1) {}; 
			\node[circle, draw] (e1) at (4,1) {};  
			\node[circle, draw] (f1) at (5,1) {};  
			\draw (a1) -- (b1) -- (c1) -- (d1) -- (e1) -- (f1); 
			\node[circle, draw, fill = black] (a2) at (0,2) {};
            \node at (a2.west) [left=8pt] {\footnotesize $u_3$};
			\node[circle, draw] (b2) at (1,2) {};
			\node[circle, draw, fill = black] (c2) at (2,2) {};
			\node[circle, draw] (d2) at (3,2) {}; 
			\node[circle, draw, fill = black] (e2) at (4,2) {};  
			\node[circle, draw] (f2) at (5,2) {};  
			\draw (a2) -- (b2) -- (c2) -- (d2) -- (e2) -- (f2);
			\node[circle, draw, fill = black] (a3) at (0,3) {};
            \node at (a3.west) [left=8pt] {\footnotesize $u_4$};
			\node[circle, draw] (b3) at (1,3) {};
			\node[circle, draw, fill = black] (c3) at (2,3) {};
			\node[circle, draw] (d3) at (3,3) {}; 
			\node[circle, draw, fill = black] (e3) at (4,3) {};  
			\node[circle, draw, fill = black] (f3) at (5,3) {};  
			\draw (a3) -- (b3) -- (c3) -- (d3) -- (e3) -- (f3);
			\node[circle, draw, fill = black] (a4) at (0,4) {};
            \node at (a4.west) [left=8pt] {\footnotesize $u_5$};
			\node[circle, draw] (b4) at (1,4) {};
			\node[circle, draw, fill = black] (c4) at (2,4) {};
			\node[circle, draw] (d4) at (3,4) {}; 
			\node[circle, draw] (e4) at (4,4) {};  
			\node[circle, draw] (f4) at (5,4) {};  
			\draw (a4) -- (b4) -- (c4) -- (d4) -- (e4) -- (f4);
			\node[circle, draw, fill = black] (a5) at (0,5) {};
            \node at (a5.west) [left=8pt] {\footnotesize $u_6$};
			\node[circle, draw] (b5) at (1,5) {};
			\node[circle, draw, fill = black] (c5) at (2,5) {};
			\node[circle, draw, fill = black] (d5) at (3,5) {}; 
			\node[circle, draw, fill = black] (e5) at (4,5) {};  
			\node[circle, draw, fill = black] (f5) at (5,5) {};  
			\draw (a5) -- (b5) -- (c5) -- (d5) -- (e5) -- (f5);
			\draw (a) -- (a1) -- (a2) -- (a3) -- (a4) -- (a5);
			\draw (b) -- (b1) -- (b2) -- (b3) -- (b4) -- (b5);
			\draw (c) -- (c1) -- (c2) -- (c3) -- (c4) -- (c5);
			\draw (d) -- (d1) -- (d2) -- (d3) -- (d4) -- (d5);
			\draw (e) -- (e1) -- (e2) -- (e3) -- (e4) -- (e5);
			\draw (f) -- (f1) -- (f2) -- (f3) -- (f4) -- (f5);

			\node[circle, draw, fill = black] (a') at (7,0) {};
            \node at (a'.west) [left=8pt] {\footnotesize $u_1$};
            \node at (a'.south) [below=8pt] {\footnotesize $v_1$};
			\node[circle, draw] (b') at (8,0) {};
            \node at (b'.south) [below=8pt] {\footnotesize $v_2$};
			\node[circle, draw, fill = black] (c') at (9,0) {};
            \node at (c'.south) [below=8pt] {\footnotesize $v_3$};
			\node[circle, draw, fill = black] (d') at (10,0) {}; 
            \node at (d'.south) [below=8pt] {\footnotesize $v_4$};
			\node[circle, draw, fill = black] (e') at (11,0) {};  
            \node at (e'.south) [below=8pt] {\footnotesize $v_5$};
			\node[circle, draw, fill = black] (f') at (12,0) {};  
            \node at (f'.south) [below=8pt] {\footnotesize $v_6$};
			\node[circle, draw, fill = black] (g') at (13,0) {};
            \node at (g'.south) [below=8pt] {\footnotesize $v_7$};
			\draw (a') -- (b') -- (c') -- (d') -- (e') -- (f') -- (g');
			\node[circle, draw, fill = black] (a1') at (7,1) {};
            \node at (a1'.west) [left=8pt] {\footnotesize $u_2$};
			\node[circle, draw] (b1') at (8,1) {$x$};
			\node[circle, draw] (c1') at (9,1) {};
			\node[circle, draw] (d1') at (10,1) {}; 
			\node[circle, draw] (e1') at (11,1) {};  
			\node[circle, draw] (f1') at (12,1) {};  
			\node[circle, draw] (g1') at (13,1) {};
			\draw (a1') -- (b1') -- (c1') -- (d1') -- (e1') -- (f1') -- (g1'); 
			\node[circle, draw, fill = black] (a2') at (7,2) {};
            \node at (a2'.west) [left=8pt] {\footnotesize $u_3$};
			\node[circle, draw] (b2') at (8,2) {};
			\node[circle, draw, fill = black] (c2') at (9,2) {};
			\node[circle, draw] (d2') at (10,2) {}; 
			\node[circle, draw, fill = black] (e2') at (11,2) {};  
			\node[circle, draw] (f2') at (12,2) {};  
			\node[circle, draw, fill = black] (g2') at (13,2) {};
			\draw (a2') -- (b2') -- (c2') -- (d2') -- (e2') -- (f2') -- (g2');
			\node[circle, draw, fill = black] (a3') at (7,3) {};
            \node at (a3'.west) [left=8pt] {\footnotesize $u_4$};
			\node[circle, draw] (b3') at (8,3) {};
			\node[circle, draw, fill = black] (c3') at (9,3) {};
			\node[circle, draw] (d3') at (10,3) {}; 
			\node[circle, draw, fill = black] (e3') at (11,3) {};  
			\node[circle, draw] (f3') at (12,3) {};  
			\node[circle, draw] (g3') at (13,3) {};
			\draw (a3') -- (b3') -- (c3') -- (d3') -- (e3') -- (f3') -- (g3');
			\node[circle, draw, fill = black] (a4') at (7,4) {};
            \node at (a4'.west) [left=8pt] {\footnotesize $u_5$};
			\node[circle, draw] (b4') at (8,4) {};
			\node[circle, draw, fill = black] (c4') at (9,4) {};
			\node[circle, draw] (d4') at (10,4) {}; 
			\node[circle, draw, fill = black] (e4') at (11,4) {};  
			\node[circle, draw, fill = black] (f4') at (12,4) {};  
			\node[circle, draw, fill = black] (g4') at (13,4) {};
			\draw (a4') -- (b4') -- (c4') -- (d4') -- (e4') -- (f4') -- (g4');
			\node[circle, draw, fill = black] (a5') at (7,5) {};
            \node at (a5'.west) [left=8pt] {\footnotesize $u_6$};
			\node[circle, draw] (b5') at (8,5) {};
			\node[circle, draw, fill = black] (c5') at (9,5) {};
			\node[circle, draw] (d5') at (10,5) {}; 
			\node[circle, draw] (e5') at (11,5) {};  
			\node[circle, draw] (f5') at (12,5) {};
			\node[circle, draw] (g5') at (13,5) {};  
			\draw (a5') -- (b5') -- (c5') -- (d5') -- (e5') -- (f5') -- (g5');
			\node[circle, draw, fill = black] (a6') at (7,6) {};
            \node at (a6'.west) [left=8pt] {\footnotesize $u_7$};
			\node[circle, draw] (b6') at (8,6) {};
			\node[circle, draw, fill = black] (c6') at (9,6) {};
			\node[circle, draw, fill = black] (d6') at (10,6) {}; 
			\node[circle, draw, fill = black] (e6') at (11,6) {};  
			\node[circle, draw, fill = black] (f6') at (12,6) {};
			\node[circle, draw, fill = black] (g6') at (13,6) {};  
			\draw (a6') -- (b6') -- (c6') -- (d6') -- (e6') -- (f6') -- (g6');

			\draw (a') -- (a1') -- (a2') -- (a3') -- (a4') -- (a5') -- (a6');
			\draw (b') -- (b1') -- (b2') -- (b3') -- (b4') -- (b5') -- (b6');
			\draw (c') -- (c1') -- (c2') -- (c3') -- (c4') -- (c5') -- (c6');
			\draw (d') -- (d1') -- (d2') -- (d3') -- (d4') -- (d5') -- (d6');
			\draw (e') -- (e1') -- (e2') -- (e3') -- (e4') -- (e5') -- (e6');
			\draw (f') -- (f1') -- (f2') -- (f3') -- (f4') -- (f5') -- (f6');
			\draw (g') -- (g1') -- (g2') -- (g3') -- (g4') -- (g5') -- (g6');

\draw (-0.5,-0.5) rectangle (0.5,0.5);
\draw (-0.5,1.5) rectangle (0.5,5.5);
\draw (1.5,-0.5) rectangle (5.5,0.5);
\draw (1.5,1.5) rectangle (5.5,5.5);

\draw[dotted] (-0.5,0.5) -- (0.5,-0.5);

\draw[dotted] (2.5,0.5) -- (1.5,-0.5);
\draw[dotted] (3.5,0.5) -- (2.5,-0.5);
\draw[dotted] (4.5,0.5) -- (3.5,-0.5);
\draw[dotted] (5.5,0.5) -- (4.5,-0.5);

\draw[dotted] (-0.5,1.5) -- (0.5,2.5);
\draw[dotted] (-0.5,2.5) -- (0.5,3.5);
\draw[dotted] (-0.5,3.5) -- (0.5,4.5);
\draw[dotted] (-0.5,4.5) -- (0.5,5.5);

\draw[dotted] (1.5,2.5) -- (2.5,1.5);
\draw[dotted] (1.5,3.5) -- (3.5,1.5);
\draw[dotted] (1.5,4.5) -- (4.5,1.5);
\draw[dotted] (1.5,5.5) -- (5.5,1.5);
\draw[dotted] (2.5,5.5) -- (5.5,2.5);
\draw[dotted] (3.5,5.5) -- (5.5,3.5);
\draw[dotted] (4.5,5.5) -- (5.5,4.5);

\draw (6.5,-0.5) rectangle (7.5,0.5);
\draw (6.5,1.5) rectangle (7.5,6.5);
\draw (8.5,-0.5) rectangle (13.5,0.5);
\draw (8.5,1.5) rectangle (13.5,6.5);

\draw[dotted] (6.5,0.5) -- (7.5,-0.5);

\draw[dotted] (9.5,0.5) -- (8.5,-0.5);
\draw[dotted] (10.5,0.5) -- (9.5,-0.5);
\draw[dotted] (11.5,0.5) -- (10.5,-0.5);
\draw[dotted] (12.5,0.5) -- (11.5,-0.5);
\draw[dotted] (13.5,0.5) -- (12.5,-0.5);

\draw[dotted] (6.5,1.5) -- (7.5,2.5);
\draw[dotted] (6.5,2.5) -- (7.5,3.5);
\draw[dotted] (6.5,3.5) -- (7.5,4.5);
\draw[dotted] (6.5,4.5) -- (7.5,5.5);
\draw[dotted] (6.5,5.5) -- (7.5,6.5);

\draw[dotted] (8.5,2.5) -- (9.5,1.5);
\draw[dotted] (8.5,3.5) -- (10.5,1.5);
\draw[dotted] (8.5,4.5) -- (11.5,1.5);
\draw[dotted] (8.5,5.5) -- (12.5,1.5);
\draw[dotted] (8.5,6.5) -- (13.5,1.5);

\draw[dotted] (9.5,6.5) -- (13.5,2.5);
\draw[dotted] (10.5,6.5) -- (13.5,3.5);
\draw[dotted] (11.5,6.5) -- (13.5,4.5);
\draw[dotted] (12.5,6.5) -- (13.5,5.5);
		\end{tikzpicture}
		\caption{vv-sets of $P_6 \cp P_6$ and $P_7 \cp P_7$.}
		\label{fig:grid}
	\end{center}
\end{figure}

By the above argument, we can have at most $\lceil \frac{|D_{pq}|}{2} \rceil$ vertices in $S \cap D_{pq}$. The diagonals $D_{pq}$ in $Q_1$ can be explicitly represented as $D_a = \{(u_k,v_l): k + l = a\}$, $a \in [i+j] \setminus \{1\}$. Thus if $i$ is even, then
$$|S \cap Q_1| = \sum_{a \in [i+j]\setminus \{1\}}|S \cap D_a| \leq 4\left(1+2+\cdots+\left\lfloor\frac{i-1}{2}\right\rfloor\right) + (j-i+3)\left\lceil\frac{i}{2}\right\rceil,$$ 
and if $i$ is odd, 
$$|S \cap Q_1| = \sum_{a \in [i+j]\setminus \{1\}}|S \cap D_a| \leq 4\left(1+2+\cdots+\left\lfloor\frac{i-1}{2}\right\rfloor\right) + (j-i+1)\left\lceil\frac{i}{2}\right\rceil.$$ 
On simplification we get, 
\begin{flalign*}
 |S \cap Q_1| \leq \begin{cases}
      \frac{i(j+1)}{2}; & \text{ $i$ even},\\
      \frac{(i+1)j}{2}; & \text{ $i$ odd}.  
    \end{cases} \label{}
    \end{flalign*}

A similar computation can be done for each $Q_t, t \in [4]$. It may be noted that by our choice of $i$ and $j$ we have, $i \leq j \leq n - j - 1 \leq n - i - 1$. Consequently we obtain the following inequalities:
\begin{flalign*}
    |S \cap Q_2| \leq \begin{cases}
      \frac{i(n-j)}{2}; & \text{ $i$ even},\\
      \frac{(i+1)(n-j-1)}{2}; & \text{ $i$ odd},  
    \end{cases} \label{} \\
    |S \cap Q_3| \leq \begin{cases}
      \frac{(n-j-1)(n-i)}{2}; & \text{ $n-j-1$ even},\\
      \frac{((n-j)(n-i-1)}{2}; & \text{ $n-j-1$ odd},  
    \end{cases} \label{} \\
    |S \cap Q_4| \leq \begin{cases}
      \frac{j(n-i)}{2}; & \text{ $j$ even},\\
      \frac{(j+1)(n-i-1)}{2}; & \text{ $j$ odd}.  
    \end{cases} \label{}  
\end{flalign*}

By the choice of $S$, every vertex must be $S$-visible from $(u_{i+1},v_{j+1})$ and hence $S \cap (X \cup Y) \subseteq \{(u_{i+1},v_1), (u_{i+1},v_n), (u_1,v_{j+1}), (u_n,v_{j+1})\}$. But, $(u_{i+1},v_1)$ can be in $S$ only if both $(u_i,v_2)$ and $(u_{i+2},v_2)$ are not in $S$ since $(u_i,v_1)$ and $(u_{i+2},v_1)$ must be $S$-visible from $(u_{i+1},v_{j+1})$. A similar argument holds for $ (u_{i+1},v_n)$, $(u_1,v_{j+1})$, and $(u_n,v_{j+1})$. Also, note that in $Q_t$, at least one end vertex of $D_{pq}$ must be in $S$ if $|D_{pq}|$ is even, and both end vertices of $D_{pq}$ must be in $S$ if $|D_{pq}|$ is odd for $S \cap Q_t$ to attain the maximum cardinality. Since $i \leq j$ we get that  
\begin{flalign*}
\text{in $Q_1$}:\  |D_{p(u_2,v_j
)}| = i-1  \text{ and } |D_{(u_i,v_2)q}| = \begin{cases}
       i - 1; &  i=j,\\
       i; & i < j.
   \end{cases}  \label{}
\end{flalign*} 
A similar computation can be done in $Q_2$, $Q_3$, and $Q_4$,  to obtain the following respective cardinalities: 
\begin{flalign*}
|D_{p(u_2,v_{j+2}
)}| & = i-1   \text{ and } |D_{(u_{i},v_{n-1})q}| = \begin{cases}
       i - 1; & i=n-j-1,\\
       i; &  i < n-j-1.
   \end{cases} \label{} \\
 |D_{p(u_{i},v_{n+1})}| & = n-j-2  \text{ and } |D_{(u_{n-1},v_{j+2})q}| = \begin{cases}
       n-j-2;& i=j,\\
       n-j-1; & i < j.
   \end{cases} \label{} \\
|D_{p(u_{i+2},v_{2})}| & = j-1  \text{ and } |D_{(u_{n-1},v_{j})q}| = \begin{cases}
       j - 1; & j=n-i-1,\\
       j; & j < n-i-1.
   \end{cases}  \label{}
\end{flalign*}
Then, since we have assumed that $S$ has the maximum number of vertices from  $\bigcup_{t=1}^{4} Q_t$, we infer that $|S \cap (X \cup Y)|$ equals $0$, $1$, or $2$, according to $n$, $i$, and $j$ being even or odd. Considering all possible combinations, we get the maximum possible cardinality as $\frac{n^2+n-2}{2}$. Therefore, $vv(P_n \cp P_n) \leq \frac{n^2+n-2}{2}$.

Now, choose $x=(u_2,v_2)$ and let 
$$S = \{(u_k,v_1):\ k \in [n]\} \cup \{(u_1,v_l): l \in [n] \backslash \{2\}\} \cup S'\,,$$ 
where $S' \subset Q_3$ is chosen by taking $\lceil\frac{|D_a|}{2}\rceil$ alternate vertices from each diagonal $D_a$ there; see Figure~\ref{fig:grid} which illustrates this construction on $P_6 \cp P_6$ and on $P_7 \cp P_7$. Then $S$ is an $x$-visibility set of cardinality $\frac{n^2+n-2}{2}$. We can conclude that $\vv(P_n \cp P_n) = \frac{n^2+n-2}{2}$.
\qed 
\end{proof}

\subsection{Square prisms}
\label{subsec:prisms}	

The proof technique used in Theorem~\ref{vv-sqauregrid} can be used to determine the vertex visibility number of square prisms. 

\begin{theorem}\label{vv-sqaurecylinder}
If $n \ge 4$, then 
$$\vv(P_n \cp C_n) = \begin{cases}
\frac{n^2 + 3}{2}; & n \equiv 1 \bmod 4,\\
\frac{n^2 + n - 2}{2}; & n \equiv 3 \bmod 4, \\
\frac{2n^2 + n}{4}; & n \equiv 0 \bmod 4,\\
\frac{2n^2 + n - 2}{4}; & n \equiv 2 \bmod 4.
\end{cases}$$    
\end{theorem}

\begin{proof}
Let $V(P_n \cp C_n) = \{(u_k,v_l): k,l \in [n]\}$. By the symmetry, a vertex $x \in V(P_n \cp C_n)$ has the same $x$-visibility number as some vertex from $\{(u_k,v_{\lceil n/2 \rceil}): k \in [\lceil n/2 \rceil]\}$, hence we may without loss of generality assume that $x = (u_{i+1},v_{\lceil n/2 \rceil})$, where $0 \leq i \leq \lfloor \frac{n-1}{2} \rfloor$. Let $X = V(^{u_{i+1}}C_n)$  and $Y = V(P_n^{v_{\lceil n/2 \rceil}})$. Set further
\begin{align*}
Q_1 & = \{(u_k,v_l):\ k < i+1, l < \lceil n/2 \rceil\}, \\
Q_2 & = \{(u_k,v_l):\ k < i+1, l > \lceil n/2 \rceil\}, \\
Q_3 & = \{(u_k,v_l):\ k > i+1, l > \lceil n/2 \rceil\}, \\
Q_4 & = \{(u_k,v_l):\ k > i+1, l < \lceil n/2 \rceil\}. 
\end{align*}
 Let $S$ be a $v_x$-set in $P_n \cp C_n$ that satisfies the properties given in Lemma \ref{lem:mdv} and has the maximum number of vertices from $\bigcup_{t=1}^{4} Q_t$ among its choices. 

 \begin{figure}[ht!]
	\begin{center}
		\begin{tikzpicture}[scale=0.5,style=thick,x=2cm,y=2cm]
			\node[circle, draw, fill = black, scale=0.75] (a) at (0,0) {};
            \node at (a.west) [left=8pt] {\footnotesize $u_1$};
            \node at (a.south) [below=8pt] {\footnotesize $v_1$};
			\node[circle, draw, fill = black, scale=0.75] (b) at (1,0) {};
            \node at (b.south) [below=8pt] {\footnotesize $v_2$};
			\node[circle, draw, fill = black, scale=0.75] (c) at (2,0) {};
            \node at (c.south) [below=8pt] {\footnotesize $v_3$};
			\node[circle, draw, fill = black, scale=0.75] (d) at (3,0) {}; 
			\node at (d.south) [below=8pt] {\footnotesize $v_4$};
			\node[circle, draw, fill = black, scale=0.75] (f) at (4,0) {};  
            \node at (f.south) [below=8pt] {\footnotesize $v_5$};
			\draw (a) -- (b) -- (c) -- (d) -- (f);

            \draw (a) .. controls (0.5,-0.5) and (3.5,-0.5) .. (f);
           
			\node[circle, draw, scale=0.75] (a1) at (0,1) {};
            \node at (a1.west) [left=8pt] {\footnotesize $u_2$};
			\node[circle, draw, scale=0.75] (b1) at (1,1){};
			\node[circle, draw, scale=0.75] (c1) at (2,1) {$x$};
			\node[circle, draw, scale=0.75] (d1) at (3,1) {}; 
			
			\node[circle, draw, scale=0.75] (f1) at (4,1) {};  
			\draw (a1) -- (b1) -- (c1) -- (d1) -- (f1); 

            \draw (a1) .. controls (0.5,0.5) and (3.5,0.5) .. (f1);

			\node[circle, draw, fill = black, scale=0.75] (a2) at (0,2) {};
            \node at (a2.west) [left=8pt] {\footnotesize $u_3$};
			\node[circle, draw, fill = black, scale=0.75] (b2) at (1,2) {};
			\node[circle, draw, scale=0.75] (c2) at (2,2)  {};
			\node[circle, draw, fill = black, scale=0.75] (d2) at (3,2) {}; 

			\node[circle, draw, fill = black, scale=0.75] (f2) at (4,2) {};  
			\draw (a2) -- (b2) -- (c2) -- (d2) -- (f2);

            \draw (a2) .. controls (0.5,1.5) and (3.5,1.5) .. (f2);

			\node[circle, draw, fill = black, scale=0.75] (a3) at (0,3) {};
            \node at (a3.west) [left=8pt] {\footnotesize $u_4$};
			\node[circle, draw,  scale=0.75] (b3) at (1,3) {};
			\node[circle, draw, scale=0.75] (c3) at (2,3) {};
			\node[circle, draw,  scale=0.75] (d3) at (3,3) {}; 
	  
			\node[circle, draw, fill = black, scale=0.75] (f3) at (4,3) {};  
			\draw (a3) -- (b3) -- (c3) -- (d3) -- (f3);

             \draw (a3) .. controls (0.5,2.5) and (3.5,2.5) .. (f3);

			\node[circle, draw, fill = black, scale=0.75] (a5) at (0,4) {};
            \node at (a5.west) [left=8pt] {\footnotesize $u_5$};
			\node[circle, draw,  scale=0.75] (b5) at (1,4) {};
			\node[circle, draw, fill = black, scale=0.75] (c5) at (2,4) {};
			\node[circle, draw,  scale=0.75] (d5) at (3,4) {}; 
 
			\node[circle, draw, fill = black, scale=0.75] (f5) at (4,4) {};  
			\draw (a5) -- (b5) -- (c5) -- (d5) -- (f5);

 \draw (a5) .. controls (0.5,3.5) and (3.5,3.5) .. (f5);

			\draw (a) -- (a1) -- (a2) -- (a3) -- (a5);
           
			\draw (b) -- (b1) -- (b2) -- (b3) -- (b5);
           
			\draw (c) -- (c1) -- (c2) -- (c3) -- (c5);
            
			\draw (d) -- (d1) -- (d2) -- (d3) -- (d5);
           
			\draw (f) -- (f1) -- (f2) -- (f3) -- (f5);

            	\node[circle, draw, fill = black, scale=0.75] (a') at (7,0) {};
                \node at (a'.west) [left=8pt] {\footnotesize $u_1$};
                \node at (a'.south) [below=8pt] {\footnotesize $v_1$};
			\node[circle, draw, fill = black, scale=0.75] (b') at (8,0) {};
            \node at (b'.south) [below=8pt] {\footnotesize $v_2$};
			\node[circle, draw, fill = black, scale=0.75] (c') at (9,0) {};
            \node at (c'.south) [below=8pt] {\footnotesize $v_3$};
			\node[circle, draw, fill = black, scale=0.75] (d') at (10,0) {}; 
            \node at (d'.south) [below=8pt] {\footnotesize $v_4$};
			\node[circle, draw, fill = black, scale=0.75] (e') at (11,0) {}; 
            \node at (e'.south) [below=8pt] {\footnotesize $v_5$};
			\node[circle, draw, fill = black, scale=0.75] (f') at (12,0) {}; 
            \node at (f'.south) [below=8pt] {\footnotesize $v_6$};
            \node[circle, draw, fill = black, scale=0.75] (h') at (13,0) {};
            \node at (h'.south) [below=8pt] {\footnotesize $v_7$};
			\draw (a') -- (b') -- (c') -- (d') -- (e') -- (f') -- (h');
            \draw (a') .. controls (7.5,-0.5) and (12.5,-0.5) .. (h');
            
			\node[circle, draw, scale=0.75] (a1') at (7,1) {};
            \node at (a1'.west) [left=8pt] {\footnotesize $u_2$};
			\node[circle, draw, scale=0.75] (b1') at (8,1) {};
			\node[circle, draw,  scale=0.75] (c1') at (9,1) {};
			\node[circle, draw, scale=0.75] (d1') at (10,1) {$x$}; 
			\node[circle, draw,  scale=0.75] (e1') at (11,1) {};  
			\node[circle, draw, scale=0.75] (f1') at (12,1) {};  
            \node[circle, draw, scale=0.75] (h1') at (13,1) {};
			\draw (a1') -- (b1') -- (c1') -- (d1') -- (e1') -- (f1') -- (h1');
            \draw (a1') .. controls (7.5,0.5) and (12.5,0.5) .. (h1');
            
			\node[circle, draw, fill = black, scale=0.75] (a2') at (7,2) {};
            \node at (a2'.west) [left=8pt] {\footnotesize $u_3$};
			\node[circle, draw, fill = black, scale=0.75] (b2') at (8,2) {};
			\node[circle, draw, fill = black, scale=0.75] (c2') at (9,2) {};
			\node[circle, draw, scale=0.75] (d2') at (10,2) {}; 
			\node[circle, draw, fill = black, scale=0.75] (e2') at (11,2) {};  
			\node[circle, draw, fill = black, scale=0.75] (f2') at (12,2) {};  	
            \node[circle, draw, fill = black, scale=0.75] (h2') at (13,2) {};
			\draw (a2') -- (b2') -- (c2') -- (d2') -- (e2') -- (f2') -- (h2');
            \draw (a2') .. controls (7.5,1.5) and (12.5,1.5) .. (h2');
			\node[circle, draw, fill = black, scale=0.75] (a3') at (7,3) {};
            \node at (a3'.west) [left=8pt] {\footnotesize $u_4$};
			\node[circle, draw, scale=0.75] (b3') at (8,3) {};
			\node[circle, draw, scale=0.75] (c3') at (9,3) {};
			\node[circle, draw, scale=0.75] (d3') at (10,3) {}; 
			\node[circle, draw, scale=0.75] (e3') at (11,3) {};  
			\node[circle, draw, scale=0.75] (f3') at (12,3) {};  
            \node[circle, draw, fill = black, scale=0.75] (h3') at (13,3) {};
			\draw (a3') -- (b3') -- (c3') -- (d3') -- (e3') -- (f3')  -- (h3');
            \draw (a3') .. controls (7.5,2.5) and (12.5,2.5) .. (h3');
            
			\node[circle, draw, fill = black, scale=0.75] (a4') at (7,4) {};
            \node at (a4'.west) [left=8pt] {\footnotesize $u_5$};
			\node[circle, draw, scale=0.75] (b4') at (8,4) {};
			\node[circle, draw, fill = black, scale=0.75] (c4') at (9,4) {};
			\node[circle, draw, scale=0.75] (d4') at (10,4) {}; 
			\node[circle, draw, fill = black, scale=0.75] (e4') at (11,4) {};  
			\node[circle, draw, scale=0.75] (f4') at (12,4) {}; 			
            \node[circle, draw, fill = black, scale=0.75] (h4') at (13,4) {};
			\draw (a4') -- (b4') -- (c4') -- (d4') -- (e4') -- (f4') -- (h4');
            \draw (a4') .. controls (7.5,3.5) and (12.5,3.5) .. (h4');
            
			\node[circle, draw, fill = black, scale=0.75] (a5') at (7,5) {};
            \node at (a5'.west) [left=8pt] {\footnotesize $u_6$};
			\node[circle, draw, scale=0.75] (b5') at (8,5) {};
			\node[circle, draw, fill = black, scale=0.75] (c5') at (9,5) {};
			\node[circle, draw, scale=0.75] (d5') at (10,5) {}; 
			\node[circle, draw, fill = black, scale=0.75] (e5') at (11,5) {};  
			\node[circle, draw, scale=0.75] (f5') at (12,5) {};
            \node[circle, draw, fill = black, scale=0.75] (h5') at (13,5) {}; 
			\draw (a5') -- (b5') -- (c5') -- (d5') -- (e5') -- (f5') -- (h5');
            \draw (a5') .. controls (7.5,4.5) and (12.5,4.5) .. (h5');
			
\node[circle, draw, fill = black, scale=0.75] (a7') at (7,6) {};
\node at (a7'.west) [left=8pt] {\footnotesize $u_7$};
			\node[circle, draw,  scale=0.75] (b7') at (8,6) {};
			\node[circle, draw, fill = black, scale=0.75] (c7') at (9,6) {};
			\node[circle, draw, scale=0.75] (d7') at (10,6) {}; 
			\node[circle, draw, fill = black, scale=0.75] (e7') at (11,6) {};  
			\node[circle, draw,  scale=0.75] (f7') at (12,6) {};		 
            \node[circle, draw, fill = black, scale=0.75] (h7') at (13,6) {}; 
			\draw (a7') -- (b7') -- (c7') -- (d7') -- (e7') -- (f7') -- (h7');
            \draw (a7') .. controls (7.5,5.5) and (12.5,5.5) .. (h7');

			\draw (a') -- (a1') -- (a2') -- (a3') -- (a4') -- (a5') -- (a7');
			\draw (b') -- (b1') -- (b2') -- (b3') -- (b4') -- (b5') -- (b7');
			\draw (c') -- (c1') -- (c2') -- (c3') -- (c4') -- (c5') -- (c7');
			\draw (d') -- (d1') -- (d2') -- (d3') -- (d4') -- (d5') -- (d7');
			\draw (e') -- (e1') -- (e2') -- (e3') -- (e4') -- (e5') -- (e7');
			\draw (f') -- (f1') -- (f2') -- (f3') -- (f4') -- (f5') -- (f7');
            \draw (h') -- (h1') -- (h2') -- (h3') -- (h4') -- (h5') -- (h7');
	\end{tikzpicture}
		\caption{vv-sets of $P_5 \cp C_5$ and $P_7 \cp C_7$.}
		\label{fig:oddcylinder}
	\end{center}
\end{figure}
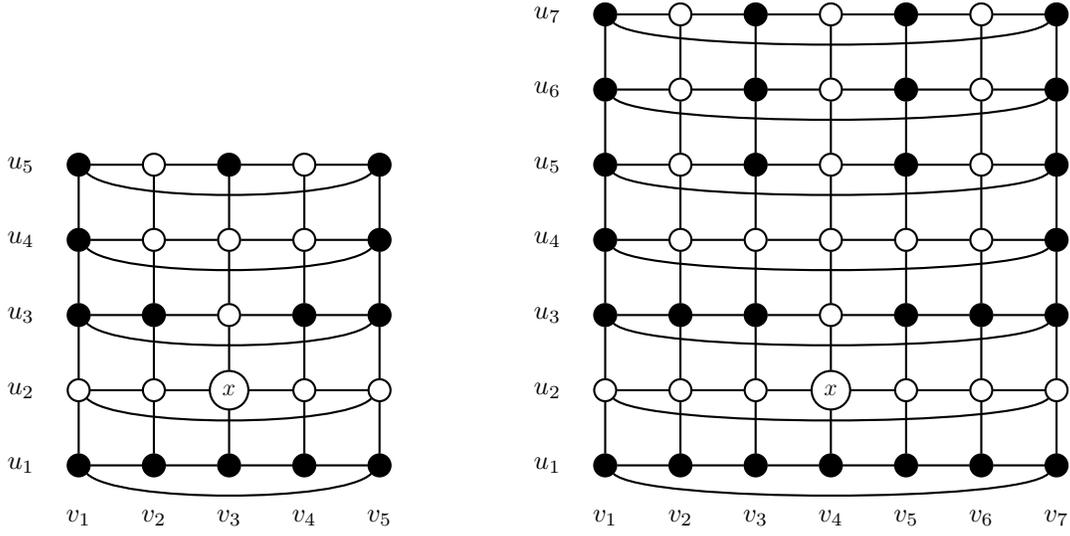

First assume $n$ is odd and $x = (u_{i+1},v_{\frac{n+1}{2}})$, where $0 \leq i \leq \frac{n-1}{2}$. For $(u_k,v_l) \in Q_1$, a shortest $x,(u_k, v_l)$-path passes through either $(u_k,v_{l+1})$ or $(u_{k+1},v_l)$, hence these two vertices cannot be simultaneously in $S$.
A similar argument is valid for each $(u_k,v_l) \in Q_t$, $t \in [4]$. 
Now, partition each $Q_t$ into diagonals $D_{pq}$, where $p, q \in Q_t$ are its end vertices as in Theorem~\ref{vv-sqauregrid}. Then we can have at most $\lceil \frac{|D_{pq}|}{2} \rceil$ vertices in $S \cap D_{pq}$.
Hence we get that
$$\text{for $t \in \{1,2\}$}:\ |S \cap Q_t| \leq 
      \begin{cases}
      \frac{i(n+1)}{4}; & i \text{ even},\\
      \frac{(i+1)(n-1)}{4}; & i \text{ odd},
      \end{cases}$$
    and
$$\text{for $t \in \{3,4\}$}:\ |S \cap Q_t| \leq 
      \begin{cases}
      \frac{(n-1)(n-i)}{4}; & n \equiv 1 \bmod 4,\\
      \frac{(n+1)(n-i-1)}{4}; & n \equiv 3 \bmod 4.
      \end{cases}$$
Thus we obtain,
\begin{flalign*}
   |S \cup \displaystyle \bigcup_{t=1}^{4} Q_t| \leq \begin{cases}
   \frac{n^2 -1}{2}; & n \equiv 1 \bmod 4,\\
   \frac{n^2+n-4}{2}; & n \equiv 3 \bmod 4. 
   \end{cases}\label{}
   \end{flalign*}
By the choice of $S$, every vertex must be $S$-visible from $x$ and hence $S \cap (X \cup Y) \subseteq \{(u_{i+1},v_1), (u_{i+1},v_n), (u_1,v_{\frac{n+1}{2}}), (u_n,v_{\frac{n+1}{2}})\}$. But, $(u_{i+1},v_1)$ can be in $S$ only if both $(u_i,v_2)$ and $(u_{i+2},v_2)$ are not in $S$ since $(u_i,v_1)$ and $(u_{i+2},v_1)$ must be $S$-visible from $x$. Similar argument holds for $(u_{i+1},v_n), (u_1,v_{\frac{n+1}{2}})$ and $(u_n,v_{\frac{n+1}{2}})$. Also, since $i \leq \frac{n-1}{2} \leq n-i-1$ we have, 
\begin{align*}
\text{in $Q_1$}: &\quad |D_{p(u_2,v_{\frac{n-1}{2}}
)}| = i-1  \text{ and } |D_{(u_i,v_2)q}| = \begin{cases}
       i - 1; &  i=\frac{n-1}{2},\\
       i; & i < \frac{n-1}{2},
   \end{cases} \\
\text{in $Q_2$}: &\quad |D_{p(u_2,v_{\frac{n+3}{2}}
)}| = i-1   \text{ and } |D_{(u_{i},v_{n-1})q}| = \begin{cases}
       i - 1; & i=\frac{n-1}{2},\\
       i; &  i < \frac{n-1}{2},
   \end{cases} \\
\text{in $Q_3$}: &\quad |D_{p(u_{i},v_{n+1})}| = \frac{n-3}{2}  \text{ and } |D_{(u_{n-1},v_{\frac{n+3}{2}})q}| = \begin{cases}
       \frac{n-3}{2};& i=\frac{n-1}{2},\\
       \frac{n-1}{2}; & i < \frac{n-1}{2},
   \end{cases} \\
\text{in $Q_4$}: &\quad |D_{p(u_{i+2},v_{2})}| = \frac{n-3}{2}  \text{ and } |D_{(u_{n-1},v_{\frac{n-1}{2}})q}| = \begin{cases}
       \frac{n-3}{2}; & i=\frac{n-1}{2},\\
       \frac{n-1}{2}; & i<\frac{n-1}{2}.
   \end{cases}  
\end{align*}
Then, since we have assumed that $S$ has the maximum number of vertices from  $\displaystyle \bigcup_{t=1}^{4} Q_t$ we get, 
\begin{flalign*}
  |S \cap (X \cup Y)|= \begin{cases}
 2; & n \equiv 1 \bmod 4,\\
 1; & n \equiv 3 \bmod 4.
\end{cases} \label{}
\end{flalign*}
Thus, $$|S| \leq \begin{cases}
\frac{n^2 +3}{2}; & n \equiv 1 \bmod 4,\\
 \frac{n^2 +n-2}{2}; & n \equiv 3 \bmod 4.
\end{cases}$$ 

\begin{figure}[ht!]
	\begin{center}
		\begin{tikzpicture}[scale=0.5,style=thick,x=2cm,y=2cm]
			\node[circle, draw, fill = black, scale=0.75] (a) at (0,0) {};
            \node at (a.west) [left=8pt] {\footnotesize $u_1$};
            \node at (a.south) [below=8pt] {\footnotesize $v_1$};
			\node[circle, draw, fill = black,  scale=0.75] (b) at (1,0) {};
            \node at (b.south) [below=8pt] {\footnotesize $v_2$};
			\node[circle, draw, fill = black, scale=0.75] (c) at (2,0) {};
            \node at (c.south) [below=8pt] {\footnotesize $v_3$};
			\node[circle, draw, fill = black, scale=0.75] (d) at (3,0) {}; 
            \node at (d.south) [below=8pt] {\footnotesize $v_4$};
			\node[circle, draw, fill = black, scale=0.75] (e) at (4,0) {}; 
            \node at (e.south) [below=8pt] {\footnotesize $v_5$};
			\node[circle, draw, fill = black, scale=0.75] (f) at (5,0) {};  
            \node at (f.south) [below=8pt] {\footnotesize $v_6$};
			\draw (a) -- (b) -- (c) -- (d) -- (e) -- (f);
            
            \draw (a) .. controls (0.5,-0.5) and (4.5,-0.5).. (f);

			\node[circle, draw, scale=0.75] (a1) at (0,1) {};
            \node at (a1.west) [left=8pt] {\footnotesize $u_2$};
			\node[circle, draw, scale=0.75] (b1) at (1,1){};
			\node[circle, draw, scale=0.75] (c1) at (2,1) {$x$};
			\node[circle, draw,  scale=0.75] (d1) at (3,1) {}; 
			\node[circle, draw, scale=0.75] (e1) at (4,1) {};  
			\node[circle, draw, scale=0.75] (f1) at (5,1) {};  
			\draw (a1) -- (b1) -- (c1) -- (d1) -- (e1) -- (f1); 

            \draw (a1) .. controls (0.5,0.5) and (4.5,0.5) .. (f1);

			\node[circle, draw, fill=black, scale=0.75] (a2) at (0,2) {};
            \node at (a2.west) [left=8pt] {\footnotesize $u_3$};
			\node[circle, draw, fill=black, scale=0.75] (b2) at (1,2) {};
			\node[circle, draw,  scale=0.75] (c2) at (2,2)  {};
			\node[circle, draw, fill = black, scale=0.75] (d2) at (3,2) {}; 
			\node[circle, draw, scale=0.75] (e2) at (4,2) {};  
			\node[circle, draw,  fill = black, scale=0.75] (f2) at (5,2) {};  
			\draw (a2) -- (b2) -- (c2) -- (d2) -- (e2) -- (f2);

            \draw (a2) .. controls (0.5,1.5) and (4.5,1.5) .. (f2);

			\node[circle, draw, fill = black, scale=0.75] (a3) at (0,3) {};
            \node at (a3.west) [left=8pt] {\footnotesize $u_4$};
			\node[circle, draw,  scale=0.75] (b3) at (1,3) {};
			\node[circle, draw, scale=0.75] (c3) at (2,3) {};
			\node[circle, draw, fill = black,  scale=0.75] (d3) at (3,3) {}; 
			\node[circle, draw,  scale=0.75] (e3) at (4,3) {};  
			\node[circle, draw, fill = black, scale=0.75] (f3) at (5,3) {};  
			\draw (a3) -- (b3) -- (c3) -- (d3) -- (e3) -- (f3);

             \draw (a3) .. controls (0.5,2.5) and (4.5,2.5) .. (f3);

			\node[circle, draw, fill = black, scale=0.75] (a4) at (0,4) {};
            \node at (a4.west) [left=8pt] {\footnotesize $u_5$};
			\node[circle, draw,  scale=0.75] (b4) at (1,4) {};
			\node[circle, draw, scale=0.75] (c4) at (2,4) {};
			\node[circle, draw, fill = black, scale=0.75] (d4) at (3,4) {}; 
			\node[circle, draw, scale=0.75] (e4) at (4,4) {};  
			\node[circle, draw, scale=0.75] (f4) at (5,4) {};  
			\draw (a4) -- (b4) -- (c4) -- (d4) -- (e4) -- (f4);

            \draw (a4) .. controls (0.5,3.5) and (4.5,3.5) .. (f4);

			\node[circle, draw, fill = black, scale=0.75] (a5) at (0,5) {};
            \node at (a5.west) [left=8pt] {\footnotesize $u_6$};
			\node[circle, draw,  scale=0.75] (b5) at (1,5) {};
			\node[circle, draw, scale=0.75] (c5) at (2,5) {};
			\node[circle, draw, fill = black, scale=0.75] (d5) at (3,5) {}; 
			\node[circle, draw, fill = black, scale=0.75] (e5) at (4,5) {};  
			\node[circle, draw, fill = black, scale=0.75] (f5) at (5,5) {};  
			\draw (a5) -- (b5) -- (c5) -- (d5) -- (e5) -- (f5);

 \draw (a5) .. controls (0.5,4.5) and (4.5,4.5) .. (f5);

			\draw (a) -- (a1) -- (a2) -- (a3) -- (a4) -- (a5);
			\draw (b) -- (b1) -- (b2) -- (b3) -- (b4) -- (b5);
			\draw (c) -- (c1) -- (c2) -- (c3) -- (c4) -- (c5);
            
			\draw (d) -- (d1) -- (d2) -- (d3) -- (d4) -- (d5);
           
			\draw (e) -- (e1) -- (e2) -- (e3) -- (e4) -- (e5);
            
			\draw (f) -- (f1) -- (f2) -- (f3) -- (f4) -- (f5);

            	\node[circle, draw, fill = black, scale=0.75] (a') at (7,0) {};
                \node at (a'.west) [left=8pt] {\footnotesize $u_1$};
            \node at (a'.south) [below=8pt] {\footnotesize $v_1$};
			\node[circle, draw, fill = black, scale=0.75] (b') at (8,0) {};
            \node at (b'.south) [below=8pt] {\footnotesize $v_2$};
			\node[circle, draw, fill = black, scale=0.75] (c') at (9,0) {};
            \node at (c'.south) [below=8pt] {\footnotesize $v_3$};
			\node[circle, draw, fill = black, scale=0.75] (d') at (10,0) {}; 
            \node at (d'.south) [below=8pt] {\footnotesize $v_4$};
			\node[circle, draw, fill = black, scale=0.75] (e') at (11,0) {};  
            \node at (e'.south) [below=8pt] {\footnotesize $v_5$};
			\node[circle, draw, fill = black, scale=0.75] (f') at (12,0) {}; 
            \node at (f'.south) [below=8pt] {\footnotesize $v_6$};
			\node[circle, draw, fill = black, scale=0.75] (g') at (13,0) {};
            \node at (g'.south) [below=8pt] {\footnotesize $v_7$};
            \node[circle, draw, fill = black, scale=0.75] (h') at (14,0) {};
            \node at (h'.south) [below=8pt] {\footnotesize $v_8$};
			\draw (a') -- (b') -- (c') -- (d') -- (e') -- (f') -- (g') -- (h');
            \draw (a') .. controls (7.5,-0.5) and (13.5, -0.5) .. (h');
			\node[circle, draw, scale=0.75] (a1') at (7,1) {};
            \node at (a1'.west) [left=8pt] {\footnotesize $u_2$};
			\node[circle, draw, scale=0.75] (b1') at (8,1) {};
			\node[circle, draw,  scale=0.75] (c1') at (9,1) {};
			\node[circle, draw, scale=0.75] (d1') at (10,1) {$x$}; 
			\node[circle, draw,  scale=0.75] (e1') at (11,1) {};  
			\node[circle, draw, scale=0.75] (f1') at (12,1) {};  
			\node[circle, draw, scale=0.75] (g1') at (13,1) {};
            \node[circle, draw, scale=0.75] (h1') at (14,1) {};
			\draw (a1') -- (b1') -- (c1') -- (d1') -- (e1') -- (f1') -- (g1') -- (h1');
            \draw (a1') .. controls (7.5,0.5) and (13.5, 0.5) .. (h1'); 
			\node[circle, draw, fill = black, scale=0.75] (a2') at (7,2) {};
            \node at (a2'.west) [left=8pt] {\footnotesize $u_3$};
			\node[circle, draw, fill = black, scale=0.75] (b2') at (8,2) {};
			\node[circle, draw, fill = black, scale=0.75] (c2') at (9,2) {};
			\node[circle, draw, scale=0.75] (d2') at (10,2) {}; 
			\node[circle, draw, fill = black, scale=0.75] (e2') at (11,2) {};  
			\node[circle, draw, scale=0.75] (f2') at (12,2) {};  
			\node[circle, draw, fill = black, scale=0.75] (g2') at (13,2) {};
            \node[circle, draw,  scale=0.75] (h2') at (14,2) {};
			\draw (a2') -- (b2') -- (c2') -- (d2') -- (e2') -- (f2') -- (g2') -- (h2');
            \draw (a2') .. controls (7.5,1.5) and (13.5, 1.5) .. (h2');
			\node[circle, draw, fill = black, scale=0.75] (a3') at (7,3) {};
            \node at (a3'.west) [left=8pt] {\footnotesize $u_4$};
			\node[circle, draw, scale=0.75] (b3') at (8,3) {};
			\node[circle, draw, scale=0.75] (c3') at (9,3) {};
			\node[circle, draw, scale=0.75] (d3') at (10,3) {}; 
			\node[circle, draw, fill = black, scale=0.75] (e3') at (11,3) {};  
			\node[circle, draw, scale=0.75] (f3') at (12,3) {};  
			\node[circle, draw, fill = black, scale=0.75] (g3') at (13,3) {};
            \node[circle, draw, scale=0.75] (h3') at (14,3) {};
			\draw (a3') -- (b3') -- (c3') -- (d3') -- (e3') -- (f3') -- (g3') -- (h3');
            \draw (a3') .. controls (7.5,2.5) and (13.5, 2.5) .. (h3');
			\node[circle, draw, fill = black, scale=0.75] (a4') at (7,4) {};
            \node at (a4'.west) [left=8pt] {\footnotesize $u_5$};
			\node[circle, draw, scale=0.75] (b4') at (8,4) {};
			\node[circle, draw, fill = black, scale=0.75] (c4') at (9,4) {};
			\node[circle, draw, scale=0.75] (d4') at (10,4) {}; 
			\node[circle, draw, fill = black, scale=0.75] (e4') at (11,4) {};  
			\node[circle, draw, scale=0.75] (f4') at (12,4) {};  
			\node[circle, draw, fill = black, scale=0.75] (g4') at (13,4) {};
            \node[circle, draw, scale=0.75] (h4') at (14,4) {};
			\draw (a4') -- (b4') -- (c4') -- (d4') -- (e4') -- (f4') -- (g4') -- (h4');
            \draw (a4') .. controls (7.5,3.5) and (13.5, 3.5) .. (h4');
			\node[circle, draw, fill = black, scale=0.75] (a5') at (7,5) {};
            \node at (a5'.west) [left=8pt] {\footnotesize $u_6$};
			\node[circle, draw, scale=0.75] (b5') at (8,5) {};
			\node[circle, draw, fill = black, scale=0.75] (c5') at (9,5) {};
			\node[circle, draw, scale=0.75] (d5') at (10,5) {}; 
			\node[circle, draw, fill = black, scale=0.75] (e5') at (11,5) {};  
			\node[circle, draw, scale=0.75] (f5') at (12,5) {};
			\node[circle, draw, fill = black, scale=0.75] (g5') at (13,5) {}; 
            \node[circle, draw, fill = black, scale=0.75] (h5') at (14,5) {}; 
			\draw (a5') -- (b5') -- (c5') -- (d5') -- (e5') -- (f5') -- (g5') -- (h5');
            \draw (a5') .. controls (7.5,4.5) and (13.5, 4.5) .. (h5');
			\node[circle, draw, fill = black, scale=0.75] (a6') at (7,6) {};
            \node at (a6'.west) [left=8pt] {\footnotesize $u_7$};
			\node[circle, draw, scale=0.75] (b6') at (8,6) {};
			\node[circle, draw, fill = black, scale=0.75] (c6') at (9,6) {};
			\node[circle, draw, scale=0.75] (d6') at (10,6) {}; 
			\node[circle, draw, fill = black, scale=0.75] (e6') at (11,6) {};  
			\node[circle, draw, scale=0.75] (f6') at (12,6) {};
			\node[circle, draw, scale=0.75] (g6') at (13,6) {}; 
            \node[circle, draw, scale=0.75] (h6') at (14,6) {}; 
			\draw (a6') -- (b6') -- (c6') -- (d6') -- (e6') -- (f6') -- (g6') -- (h6');
            \draw (a6') .. controls (7.5,5.5) and (13.5, 5.5) .. (h6');
\node[circle, draw, fill = black, scale=0.75] (a7') at (7,7) {};
\node at (a7'.west) [left=8pt] {\footnotesize $u_8$};
			\node[circle, draw,  scale=0.75] (b7') at (8,7) {};
			\node[circle, draw, fill = black, scale=0.75] (c7') at (9,7) {};
			\node[circle, draw, scale=0.75] (d7') at (10,7) {}; 
			\node[circle, draw, fill = black, scale=0.75] (e7') at (11,7) {};  
			\node[circle, draw, fill = black, scale=0.75] (f7') at (12,7) {};
			\node[circle, draw, fill = black, scale=0.75] (g7') at (13,7) {}; 
            \node[circle, draw, fill = black, scale=0.75] (h7') at (14,7) {}; 
			\draw (a7') -- (b7') -- (c7') -- (d7') -- (e7') -- (f7') -- (g7') -- (h7');
\draw (a7') .. controls (7.5,6.5) and (13.5, 6.5) .. (h7');
            
			\draw (a') -- (a1') -- (a2') -- (a3') -- (a4') -- (a5') -- (a6') -- (a7');
			\draw (b') -- (b1') -- (b2') -- (b3') -- (b4') -- (b5') -- (b6') -- (b7');
			\draw (c') -- (c1') -- (c2') -- (c3') -- (c4') -- (c5') -- (c6') -- (c7');
			\draw (d') -- (d1') -- (d2') -- (d3') -- (d4') -- (d5') -- (d6') -- (d7');
			\draw (e') -- (e1') -- (e2') -- (e3') -- (e4') -- (e5') -- (e6') -- (e7');
			\draw (f') -- (f1') -- (f2') -- (f3') -- (f4') -- (f5') -- (f6') -- (f7');
			\draw (g') -- (g1') -- (g2') -- (g3') -- (g4') -- (g5') -- (g6') -- (g7');
            \draw (h') -- (h1') -- (h2') -- (h3') -- (h4') -- (h5') -- (h6') -- (h7');
	\end{tikzpicture}
		\caption{vv-sets of $P_6 \cp C_6$ and $P_8 \cp C_8$.}
		\label{fig:evencylinder}
	\end{center}
\end{figure}
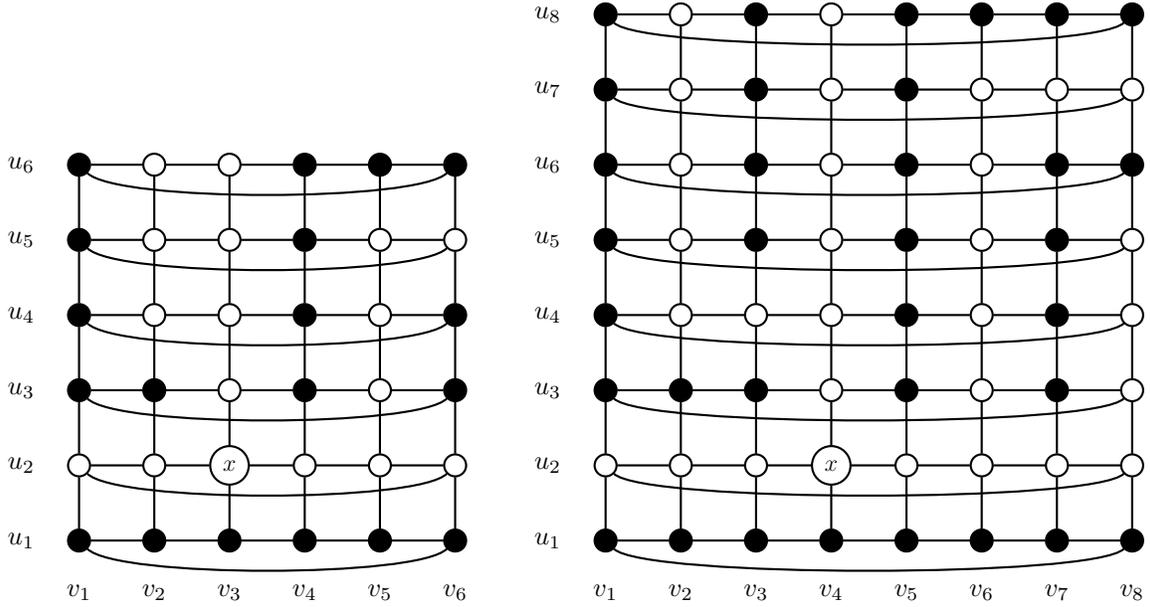
Now, assume $n$ is even and $x = (u_{i+1},v_{\frac{n}{2}})$, where $0 \leq i \leq \frac{n-2}{2}$. For each $(u_k,v_l) \in Q_t$, we know that the set $N((u_k,v_l)) \cap I(x,(u_k, v_l))$ has cardinality at most three and has a vertex which is not in $S$. Particularly for $(u_k,v_l) \in Q_1$, this vertex is either $(u_k,v_{l+1})$ or $(u_{k+1},v_l)$. Now, partition each $Q_t$ into diagonals $D_{pq}$, where $p, q \in Q_t$ are its end vertices. Then we can have at most $\sum_{pq} \lceil |D_{pq}| / 2 \rceil$ vertices in $S \cap \bigcup_{t=1}^{4} Q_t$. Since we have
$$|S \cap Q_1| + |S \cap Q_2| \leq 
 \frac{(i+1)(n-2)}{4} + \frac{(i+1)n}{4},$$
we obtain that for $n \equiv 0 \bmod 4$,
\begin{align*}
|S \cup \displaystyle \bigcup_{t=1}^{4} Q_t| & = |S \cap Q_1| + |S \cap Q_2| + |S \cap Q_3| + |S \cap Q_4| \\
& \le\frac{(i+1)(n-2)}{4} + \frac{(i+1)n}{4} + \frac{n(n-i)}{4} + \frac{n(n-i-1)}{4} \\
& = \frac{2n^2 + n-4}{4},
\end{align*}    
and for $n \equiv 2 \bmod 4$,
\begin{align*}
|S \cup \displaystyle \bigcup_{t=1}^{4} Q_t| & = |S \cap Q_1| + |S \cap Q_2| + |S \cap Q_3| + |S \cap Q_4| \\
& \le\frac{(i+1)(n-2)}{4} + \frac{(i+1)n}{4} + \frac{(n+2)(n-i-1)}{4} + \frac{(n - 2)(n-i)}{4} \\
& = \frac{2n^2 +n-6}{4}\,.
\end{align*} 
Again, by the choice of $S$, every vertex must be $S$-visible from $x$ and hence, $S \cap (X \cup Y) \subseteq \{(u_{i+1},v_1), (u_{i+1},v_{n-1}), (u_{i+1},v_n), (u_1,v_{\frac{n}{2}}), (u_n,v_{\frac{n}{2}})\}$.
Since $i \leq \frac{n-2}{2} \leq n-i-1$ we have
$$|D_{p(u_2,v_\frac{n-2}{2}
)}| = |D_{p(u_2,v_{\frac{n+2}{2}}
)}| = i-1$$
and 
$$|D_{(u_i,v_2)q}| = |D_{(u_{i},v_{n-2})q}| = \begin{cases}
       i - 1; &  i=\frac{n-2}{2},\\
       i; & i < \frac{n-2}{2},
   \end{cases}$$
$$|D_{(u_{i},v_{n-1})q}| = \begin{cases}
       i - 1; & i=\frac{n}{2},\\
       i; &  i < \frac{n}{2}
   \end{cases}$$
$$|D_{p(u_{i},v_{n+1})}| = \frac{n-2}{2}  \text{ and } |D_{(u_{n-1},v_{\frac{n+2}{2}})q}| = \begin{cases}
       \frac{n-2}{2};& i=\frac{n-2}{2},\\
       \frac{n}{2}; & i < \frac{n-2}{2},
   \end{cases}$$
$$|D_{p(u_{i+2},v_{n-2})}| = |D_{p(u_{i+2},v_{2})}| = \frac{n-4}{2}  \text{ and } |D_{(u_{n-1},v_{\frac{n-2}{2}})q}| = \begin{cases}
       \frac{n-4}{2}; & i=\frac{n}{2},\\
       \frac{n-2}{2}; & i<\frac{n}{2}.
   \end{cases}$$
Then, since we have assumed that $S$ has the maximum number of vertices from  $\displaystyle \bigcup_{t=1}^{4} Q_t$, we get, 
\begin{flalign*}
  |S \cap (X \cup Y)|= 1. \label{}
\end{flalign*}
Thus, $$|S| \leq \begin{cases}
 \frac{2n^2 +n}{4}; & n \equiv 0 \bmod 4,\\
 \frac{2n^2 +n-2}{4}; & n \equiv 2 \bmod 4.
\end{cases}$$

Now, choose $x = (u_2,v_{\lceil \frac{n}{2} \rceil})$ and let 
$$R = \{(u_1,v_l): l \in [n] \backslash \{\lceil \frac{n}{2} \rceil\}\} \cup R' \cup R''\,,$$ 
where $R' \subset Q_3$ and $R'' \subset Q_4$ are chosen by taking $\lceil\frac{|D_{pq}|}{2}\rceil$ alternate vertices from each diagonal $D_{pq}$ there; see Figures~\ref{fig:oddcylinder} and~\ref{fig:evencylinder} which illustrates this construction on $P_5 \cp C_5$, $P_6 \cp C_6$, $P_7 \cp C_7$ and $P_8 \cp P_8$. Then $R\cup \{(u_1,v_{\frac{n+1}{2}}), (u_n,v_{\frac{n+1}{2}})\}$ is an $x$-visibility set of cardinality $\frac{n^2 +3}{2}$ for $n \equiv 1 \bmod 4$ and $R\cup \{(u_1,v_{\frac{n+1}{2}})\}$ is an $x$-visibility set of cardinality $\frac{n^2 +n-2}{2}$ for $n \equiv 3 \bmod 4$. Also, $R\cup \{(u_1,v_{\frac{n+1}{2}})\}$ is an $x$-visibility set of cardinality $\frac{2n^2 +n}{4}$ for $n \equiv 0 \bmod 4$ and of cardinality $\frac{2n^2 +n-2}{4}$ for $n \equiv 2 \bmod 4$. Hence proved.
\qed
\end{proof}

\subsection{Square toruses}
\label{subsec:toruses}	

\begin{theorem}\label{vv-sqauretorus}
If $n \ge 4$, then 
$$\vv(C_n \cp C_n) = \begin{cases}
\frac{n^2 - 1}{2}; & n \equiv 1 \bmod 4,\\
\frac{n^2 + 3}{2}; & n \equiv 3 \bmod 4,\\
\frac{n^2 + 2}{2}; & n\ {\rm even}.
\end{cases}$$    
\end{theorem}

\begin{proof}
   Let $V(C_n \cp C_n) = \{(u_k,v_l): i,j \in [n]\}$. By the symmetry, all vertices of $C_n \cp C_n$ have the same visibility number, and hence we may without loss of generality assume that $x = (u_{\lceil n/2 \rceil},v_{\lceil n/2 \rceil})$. Let $X = V(^{u_{\lceil n/2 \rceil}}C_n)$  and $Y = V(C_n^{v_{\lceil n/2 \rceil}})$. Set further
\begin{align*}
Q_1 & = \{(u_k,v_l): k,l < {\lceil n/2 \rceil}\}, \\
Q_2 & = \{(u_k,v_l): k < {\lceil n/2 \rceil}, l > {\lceil n/2 \rceil}\}, \\
Q_3 & = \{(u_k,v_l): k, l > {\lceil n/2 \rceil}\}, \\
Q_4 & = \{(u_k,v_l): k > {\lceil n/2 \rceil}, l < {\lceil n/2 \rceil}\}. 
\end{align*}
Let $S$ be a $v_x$-set in $C_n \cp C_n$ that satisfies the properties given in Lemma \ref{lem:mdv} and has the maximum number of vertices from $\bigcup_{t=1}^{4} Q_t$ among its choices. 

  \begin{figure}[ht!]
	\begin{center}
		\begin{tikzpicture}[scale=0.5,style=thick,x=2cm,y=2cm]
			\node[circle, draw, fill = black, scale=0.75] (a) at (0,0) {};
            \node at (a.west) [left=8pt] {\footnotesize $u_1$};
            \node at (a.south) [below=8pt] {\footnotesize $v_1$};
			\node[circle, draw, fill = black, scale=0.75] (b) at (1,0) {};
            \node at (b.south) [below=8pt] {\footnotesize $v_2$};
			\node[circle, draw, scale=0.75] (c) at (2,0) {};
            \node at (c.south) [below=8pt] {\footnotesize $v_3$};
			\node[circle, draw, fill = black, scale=0.75] (d) at (3,0) {};
            \node at (d.south) [below=8pt] {\footnotesize $v_4$};
			\node[circle, draw, fill = black, scale=0.75] (f) at (4,0) {};
            \node at (f.south) [below=8pt] {\footnotesize $v_5$};
			\draw (a) -- (b) -- (c) -- (d) -- (f);
            
            \draw (a) .. controls (0.5,-0.5) and (3.5,-0.5).. (f);

			\node[circle, draw, scale=0.75] (a1) at (0,1) {};
            \node at (a1.west) [left=8pt] {\footnotesize $u_2$};
			\node[circle, draw, fill = black, scale=0.75] (b1) at (1,1){};
			\node[circle, draw, scale=0.75] (c1) at (2,1) {};
			\node[circle, draw, fill = black, scale=0.75] (d1) at (3,1) {}; 
			
			\node[circle, draw, scale=0.75] (f1) at (4,1) {};  
			\draw (a1) -- (b1) -- (c1) -- (d1) -- (f1); 

            \draw (a1) .. controls (0.5,0.5) and (3.5,0.5) .. (f1);

			\node[circle, draw, scale=0.75] (a2) at (0,2) {};
            \node at (a2.west) [left=8pt] {\footnotesize $u_3$};
			\node[circle, draw, scale=0.75] (b2) at (1,2) {};
			\node[circle, draw, scale=0.75] (c2) at (2,2)  {$x$};
			\node[circle, draw, scale=0.75] (d2) at (3,2) {}; 

			\node[circle, draw, scale=0.75] (f2) at (4,2) {};  
			\draw (a2) -- (b2) -- (c2) -- (d2) -- (f2);

            \draw (a2) .. controls (0.5,1.5) and (3.5,1.5) .. (f2);

			\node[circle, draw, scale=0.75] (a3) at (0,3) {};
            \node at (a3.west) [left=8pt] {\footnotesize $u_4$};
			\node[circle, draw, fill = black, scale=0.75] (b3) at (1,3) {};
			\node[circle, draw, scale=0.75] (c3) at (2,3) {};
			\node[circle, draw, fill = black, scale=0.75] (d3) at (3,3) {}; 
	  
			\node[circle, draw, scale=0.75] (f3) at (4,3) {};  
			\draw (a3) -- (b3) -- (c3) -- (d3) -- (f3);

             \draw (a3) .. controls (0.5,2.5) and (3.5,2.5) .. (f3);

			\node[circle, draw, fill = black, scale=0.75] (a5) at (0,4) {};
            \node at (a5.west) [left=8pt] {\footnotesize $u_5$};
			\node[circle, draw, fill = black, scale=0.75] (b5) at (1,4) {};
			\node[circle, draw, scale=0.75] (c5) at (2,4) {};
			\node[circle, draw, fill = black, scale=0.75] (d5) at (3,4) {}; 
 
			\node[circle, draw, fill = black, scale=0.75] (f5) at (4,4) {};  
			\draw (a5) -- (b5) -- (c5) -- (d5) -- (f5);

 \draw (a5) .. controls (0.5,3.5) and (3.5,3.5) .. (f5);

			\draw (a) -- (a1) -- (a2) -- (a3) -- (a5);
            \draw (a) .. controls (-0.5,0.5) and (-0.5,3.5) .. (a5);
			\draw (b) -- (b1) -- (b2) -- (b3) -- (b5);
            \draw (b) .. controls (0.5,0.5) and (0.5,3.5) .. (b5);
			\draw (c) -- (c1) -- (c2) -- (c3) -- (c5);
             \draw (c) .. controls (1.5,0.5) and (1.5,3.5) .. (c5);
			\draw (d) -- (d1) -- (d2) -- (d3) -- (d5);
            \draw (d) .. controls (2.5,0.5) and (2.5,3.5) .. (d5);

			\draw (f) -- (f1) -- (f2) -- (f3) -- (f5);
             \draw (f) .. controls (3.5,0.5) and (3.5,3.5) .. (f5);

            	\node[circle, draw, fill = black, scale=0.75] (a') at (7,0) {};
                \node at (a'.west) [left=8pt] {\footnotesize $u_1$};
            \node at (a'.south) [below=8pt] {\footnotesize $v_1$};
			\node[circle, draw, fill = black, scale=0.75] (b') at (8,0) {};
            \node at (b'.south) [below=8pt] {\footnotesize $v_2$};
			\node[circle, draw, fill = black, scale=0.75] (c') at (9,0) {};
            \node at (c'.south) [below=8pt] {\footnotesize $v_3$};
			\node[circle, draw, scale=0.75] (d') at (10,0) {};
            \node at (d'.south) [below=8pt] {\footnotesize $v_4$};
			\node[circle, draw, fill = black, scale=0.75] (e') at (11,0) {};  
            \node at (e'.south) [below=8pt] {\footnotesize $v_5$};
			\node[circle, draw, fill = black, scale=0.75] (f') at (12,0) {}; 
            \node at (f'.south) [below=8pt] {\footnotesize $v_6$};
			
            \node[circle, draw, fill = black, scale=0.75] (h') at (13,0) {};
            \node at (h'.south) [below=8pt] {\footnotesize $v_7$};
			\draw (a') -- (b') -- (c') -- (d') -- (e') -- (f') -- (h');
            \draw (a') .. controls (7.5,-0.5) and (12.5, -0.5).. (h');
			\node[circle, draw, scale=0.75] (a1') at (7,1) {};
            \node at (a1'.west) [left=8pt] {\footnotesize $u_2$};
			\node[circle, draw, scale=0.75] (b1') at (8,1) {};
			\node[circle, draw, fill = black, scale=0.75] (c1') at (9,1) {};
			\node[circle, draw, scale=0.75] (d1') at (10,1) {}; 
			\node[circle, draw, fill = black, scale=0.75] (e1') at (11,1) {};  
			\node[circle, draw, scale=0.75] (f1') at (12,1) {};  
			
            \node[circle, draw, scale=0.75] (h1') at (13,1) {};
			\draw (a1') -- (b1') -- (c1') -- (d1') -- (e1') -- (f1') -- (h1'); 
            \draw (a1') .. controls (7.5,0.5) and (12.5, 0.5).. (h1');
			\node[circle, draw, fill = black, scale=0.75] (a2') at (7,2) {};
            \node at (a2'.west) [left=8pt] {\footnotesize $u_3$};
			\node[circle, draw, scale=0.75] (b2') at (8,2) {};
			\node[circle, draw, fill = black, scale=0.75] (c2') at (9,2) {};
			\node[circle, draw, scale=0.75] (d2') at (10,2) {}; 
			\node[circle, draw, fill = black, scale=0.75] (e2') at (11,2) {};  
			\node[circle, draw, scale=0.75] (f2') at (12,2) {};  
			
            \node[circle, draw, fill = black, scale=0.75] (h2') at (13,2) {};
			\draw (a2') -- (b2') -- (c2') -- (d2') -- (e2') -- (f2') -- (h2');
            \draw (a2') .. controls (7.5,1.5) and (12.5, 1.5).. (h2');
			\node[circle, draw, fill = black, scale=0.75] (a3') at (7,3) {};
            \node at (a3'.west) [left=8pt] {\footnotesize $u_4$};
			\node[circle, draw, scale=0.75] (b3') at (8,3) {};
			\node[circle, draw, scale=0.75] (c3') at (9,3) {};
			\node[circle, draw, scale=0.75] (d3') at (10,3) {$x$}; 
			\node[circle, draw, scale=0.75] (e3') at (11,3) {};  
			\node[circle, draw, scale=0.75] (f3') at (12,3) {};  
			
            \node[circle, draw, fill = black, scale=0.75] (h3') at (13,3) {};
			\draw (a3') -- (b3') -- (c3') -- (d3') -- (e3') -- (f3')  -- (h3');
            \draw (a3') .. controls (7.5,2.5) and (12.5, 2.5).. (h3');
			\node[circle, draw, fill = black, scale=0.75] (a4') at (7,4) {};
            \node at (a4'.west) [left=8pt] {\footnotesize $u_5$};
			\node[circle, draw, scale=0.75] (b4') at (8,4) {};
			\node[circle, draw, fill = black, scale=0.75] (c4') at (9,4) {};
			\node[circle, draw, scale=0.75] (d4') at (10,4) {}; 
			\node[circle, draw, fill = black, scale=0.75] (e4') at (11,4) {};  
			\node[circle, draw, scale=0.75] (f4') at (12,4) {};  
			
            \node[circle, draw, fill = black, scale=0.75] (h4') at (13,4) {};
			\draw (a4') -- (b4') -- (c4') -- (d4') -- (e4') -- (f4') -- (h4');
            \draw (a4') .. controls (7.5,3.5) and (12.5, 3.5).. (h4');
			\node[circle, draw, scale=0.75] (a5') at (7,5) {};
            \node at (a5'.west) [left=8pt] {\footnotesize $u_6$};
			\node[circle, draw, scale=0.75] (b5') at (8,5) {};
			\node[circle, draw, fill = black, scale=0.75] (c5') at (9,5) {};
			\node[circle, draw, scale=0.75] (d5') at (10,5) {}; 
			\node[circle, draw, fill = black, scale=0.75] (e5') at (11,5) {};  
			\node[circle, draw, scale=0.75] (f5') at (12,5) {};
            \node[circle, draw, scale=0.75] (h5') at (13,5) {}; 
			\draw (a5') -- (b5') -- (c5') -- (d5') -- (e5') -- (f5') -- (h5');
            \draw (a5') .. controls (7.5,4.5) and (12.5, 4.5).. (h5');
			
\node[circle, draw, fill = black, scale=0.75] (a7') at (7,6) {};
\node at (a7'.west) [left=8pt] {\footnotesize $u_7$};
			\node[circle, draw, fill = black, scale=0.75] (b7') at (8,6) {};
			\node[circle, draw, fill = black, scale=0.75] (c7') at (9,6) {};
			\node[circle, draw, scale=0.75] (d7') at (10,6) {}; 
			\node[circle, draw, fill = black, scale=0.75] (e7') at (11,6) {};  
			\node[circle, draw, fill = black, scale=0.75] (f7') at (12,6) {};
			 
            \node[circle, draw, fill = black, scale=0.75] (h7') at (13,6) {}; 
			\draw (a7') -- (b7') -- (c7') -- (d7') -- (e7') -- (f7') -- (h7');
            \draw (a7') .. controls (7.5,5.5) and (12.5, 5.5).. (h7');

			\draw (a') -- (a1') -- (a2') -- (a3') -- (a4') -- (a5') -- (a7');
            \draw (a') .. controls (6.5,0.5) and (6.5, 5.5).. (a7');
			\draw (b') -- (b1') -- (b2') -- (b3') -- (b4') -- (b5') -- (b7');
            \draw (b') .. controls (7.5,0.5) and (7.5, 5.5).. (b7');
			\draw (c') -- (c1') -- (c2') -- (c3') -- (c4') -- (c5') -- (c7');
            \draw (c') .. controls (8.5,0.5) and (8.5, 5.5).. (c7');
			\draw (d') -- (d1') -- (d2') -- (d3') -- (d4') -- (d5') -- (d7');
            \draw (d') .. controls (9.5,0.5) and (9.5, 5.5).. (d7');
			\draw (e') -- (e1') -- (e2') -- (e3') -- (e4') -- (e5') -- (e7');
            \draw (e') .. controls (10.5,0.5) and (10.5, 5.5).. (e7');
			\draw (f') -- (f1') -- (f2') -- (f3') -- (f4') -- (f5') -- (f7');
            \draw (f') .. controls (11.5,0.5) and (11.5, 5.5).. (f7');
            \draw (h') -- (h1') -- (h2') -- (h3') -- (h4') -- (h5') -- (h7');
            \draw (h') .. controls (12.5,0.5) and (12.5, 5.5).. (h7');
	\end{tikzpicture}
		\caption{vv-sets of $C_5 \cp C_5$ and $C_7 \cp C_7$.}
		\label{fig:oddtorus}
	\end{center}
\end{figure}
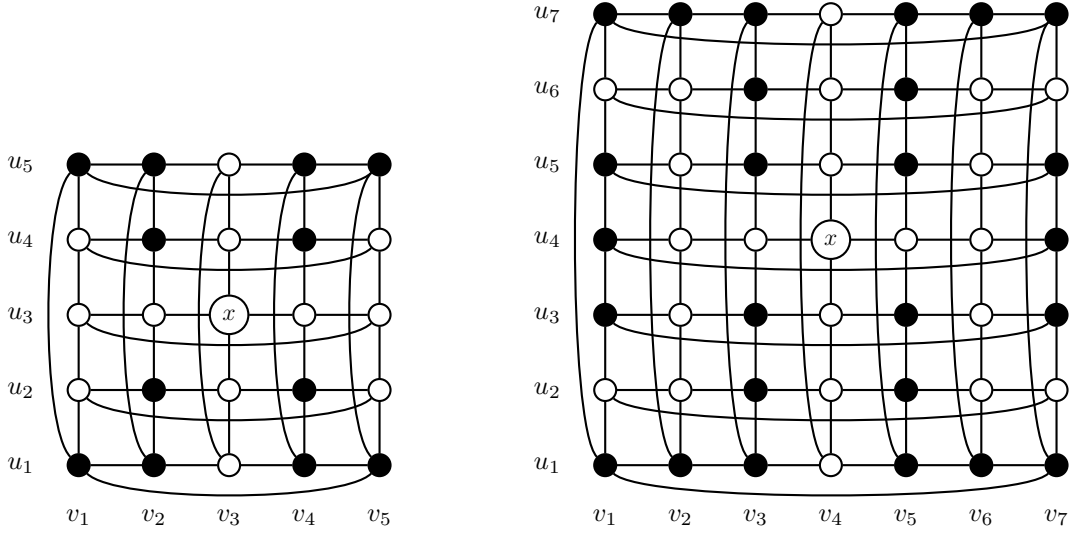

First assume $n$ is odd and $x = (u_{\frac{n+1}{2}},v_{\frac{n+1}{2}})$. For $(u_k,v_l) \in Q_1$, a shortest $x,(u_k, v_l)$-path contains either $(u_k,v_{l+1})$ or $(u_{k+1},v_l)$, hence these two vertices cannot simultaneously belong to $S$.
A similar argument is valid for each $(u_k,v_l) \in Q_t$, $t \in [4]$. 
Now, partition each $Q_t$ into diagonals $D_{pq}$, where $p, q \in Q_t$ are its end vertices as in Theorem~\ref{vv-sqauregrid}. Then we can have at most $\lceil \frac{|D_{pq}|}{2} \rceil$ vertices in $S \cap D_{pq}$.
Thus we get,
   $$|S \cap Q_t| \leq 
      \frac{n^2-1}{8}$$
for each $t \in [4]$ and hence, 
\begin{flalign*}
   |S \cup \displaystyle \bigcup_{t=1}^{4} Q_t| \leq \frac{n^2 -1}{2}. \label{}
   \end{flalign*}
By the choice of $S$, every vertex must be $S$-visible from $x$ and hence $S \cap (X \cup Y) \subseteq \{(u_{\frac{n+1}{2}},v_1), (u_{\frac{n+1}{2}},v_n), (u_1,v_{\frac{n+1}{2}}), (u_n,v_{\frac{n+1}{2}})\}$. But, $(u_{\frac{n+1}{2}},v_1)$ can be in $S$ only if both $(u_\frac{n-1}{2},v_2)$ and $(u_{\frac{n+3}{2}},v_2)$ are not in $S$ since $(u_\frac{n-1}{2},v_1)$ and $(u_{\frac{n+3}{2}},v_1)$ must be $S$-visible from $x$. Similar argument holds for $(u_{\frac{n+1}{2}},v_n)$, $(u_1,v_{\frac{n+1}{2}})$, and  $(u_n,v_{\frac{n+1}{2}})$. Also, each of the diagonals $D_{(u_\frac{n-1}{2},v_2)(u_2,v_\frac{n-1}{2})}, D_{(u_\frac{n-1}{2},v_{n-1})(u_2,v_\frac{n+3}{2})}, D_{(u_{n-1},v_\frac{n+3}{2})(u_\frac{n+3}{2},v_{n-1})}$, and $D_{(u_{n-1},v_\frac{n-1}{2})(u_\frac{n+3}{2},v_{2})}$ is of cardinality $\frac{n-3}{2}$. 
Then, since we have assumed that $S$ has the maximum number of vertices from  $\bigcup_{t=1}^{4} Q_t$, we get, 
\begin{flalign*}
  |S \cap (X \cup Y)|= \begin{cases}
 0; & n \equiv 1 \bmod 4,\\
 2; & n \equiv 3 \bmod 4.
\end{cases} \label{}
\end{flalign*}
Thus, $$|S| \leq \begin{cases}
\frac{n^2 -1}{2}; & n \equiv 1 \bmod 4,\\
 \frac{n^2 +3}{2}; & n \equiv 3 \bmod 4.
\end{cases}$$ 

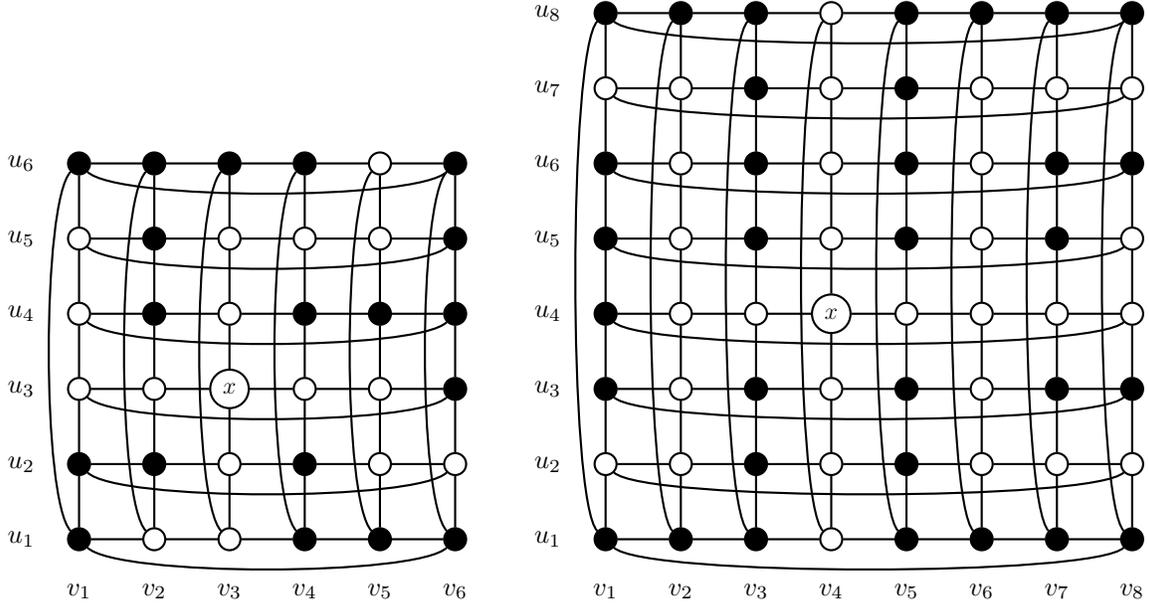
\begin{figure}[ht!]
	\begin{center}
		\begin{tikzpicture}[scale=0.5,style=thick,x=2cm,y=2cm]
			\node[circle, draw, fill = black, scale=0.75] (a) at (0,0) {};
            \node at (a.west) [left=8pt] {\footnotesize $u_1$};
            \node at (a.south) [below=8pt] {\footnotesize $v_1$};
			\node[circle, draw,  scale=0.75] (b) at (1,0) {};
            \node at (b.south) [below=8pt] {\footnotesize $v_2$};
			\node[circle, draw, scale=0.75] (c) at (2,0) {};
            \node at (c.south) [below=8pt] {\footnotesize $v_3$};
			\node[circle, draw, fill = black, scale=0.75] (d) at (3,0) {}; 
            \node at (d.south) [below=8pt] {\footnotesize $v_4$};
			\node[circle, draw, fill = black, scale=0.75] (e) at (4,0) {}; 
            \node at (e.south) [below=8pt] {\footnotesize $v_5$};
			\node[circle, draw, fill = black, scale=0.75] (f) at (5,0) {};  
            \node at (f.south) [below=8pt] {\footnotesize $v_6$};
			\draw (a) -- (b) -- (c) -- (d) -- (e) -- (f);
            
            \draw (a) .. controls (0.5,-0.5) and (4.5,-0.5).. (f);

			\node[circle, draw, fill = black, scale=0.75] (a1) at (0,1) {};
            \node at (a1.west) [left=8pt] {\footnotesize $u_2$};
			\node[circle, draw, fill = black, scale=0.75] (b1) at (1,1){};
			\node[circle, draw, scale=0.75] (c1) at (2,1) {};
			\node[circle, draw, fill = black, scale=0.75] (d1) at (3,1) {}; 
			\node[circle, draw, scale=0.75] (e1) at (4,1) {};  
			\node[circle, draw, scale=0.75] (f1) at (5,1) {};  
			\draw (a1) -- (b1) -- (c1) -- (d1) -- (e1) -- (f1); 

            \draw (a1) .. controls (0.5,0.5) and (4.5,0.5) .. (f1);

			\node[circle, draw, scale=0.75] (a2) at (0,2) {};
            \node at (a2.west) [left=8pt] {\footnotesize $u_3$};
			\node[circle, draw, scale=0.75] (b2) at (1,2) {};
			\node[circle, draw, scale=0.75] (c2) at (2,2)  {$x$};
			\node[circle, draw, scale=0.75] (d2) at (3,2) {}; 
			\node[circle, draw, scale=0.75] (e2) at (4,2) {};  
			\node[circle, draw, fill = black, scale=0.75] (f2) at (5,2) {};  
			\draw (a2) -- (b2) -- (c2) -- (d2) -- (e2) -- (f2);

            \draw (a2) .. controls (0.5,1.5) and (4.5,1.5) .. (f2);

			\node[circle, draw, scale=0.75] (a3) at (0,3) {};
            \node at (a3.west) [left=8pt] {\footnotesize $u_4$};
			\node[circle, draw, fill = black, scale=0.75] (b3) at (1,3) {};
			\node[circle, draw, scale=0.75] (c3) at (2,3) {};
			\node[circle, draw, fill = black, scale=0.75] (d3) at (3,3) {}; 
			\node[circle, draw, fill = black, scale=0.75] (e3) at (4,3) {};  
			\node[circle, draw, fill = black, scale=0.75] (f3) at (5,3) {};  
			\draw (a3) -- (b3) -- (c3) -- (d3) -- (e3) -- (f3);

             \draw (a3) .. controls (0.5,2.5) and (4.5,2.5) .. (f3);

			\node[circle, draw, scale=0.75] (a4) at (0,4) {};
            \node at (a4.west) [left=8pt] {\footnotesize $u_5$};
			\node[circle, draw, fill = black, scale=0.75] (b4) at (1,4) {};
			\node[circle, draw, scale=0.75] (c4) at (2,4) {};
			\node[circle, draw, scale=0.75] (d4) at (3,4) {}; 
			\node[circle, draw, scale=0.75] (e4) at (4,4) {};  
			\node[circle, draw, fill = black, scale=0.75] (f4) at (5,4) {};  
			\draw (a4) -- (b4) -- (c4) -- (d4) -- (e4) -- (f4);

            \draw (a4) .. controls (0.5,3.5) and (4.5,3.5) .. (f4);

			\node[circle, draw, fill = black, scale=0.75] (a5) at (0,5) {};
            \node at (a5.west) [left=8pt] {\footnotesize $u_6$};
			\node[circle, draw, fill = black, scale=0.75] (b5) at (1,5) {};
			\node[circle, draw, , fill = black, scale=0.75] (c5) at (2,5) {};
			\node[circle, draw, fill = black, scale=0.75] (d5) at (3,5) {}; 
			\node[circle, draw,  scale=0.75] (e5) at (4,5) {};  
			\node[circle, draw, fill = black, scale=0.75] (f5) at (5,5) {};  
			\draw (a5) -- (b5) -- (c5) -- (d5) -- (e5) -- (f5);

 \draw (a5) .. controls (0.5,4.5) and (4.5,4.5) .. (f5);

			\draw (a) -- (a1) -- (a2) -- (a3) -- (a4) -- (a5);
            \draw (a) .. controls (-0.5,0.5) and (-0.5, 4.5).. (a5);
			\draw (b) -- (b1) -- (b2) -- (b3) -- (b4) -- (b5);
            \draw (b) .. controls (0.5,0.5) and (0.5, 4.5) .. (b5);
			\draw (c) -- (c1) -- (c2) -- (c3) -- (c4) -- (c5);
             \draw (c) .. controls (1.5,0.5) and (1.5, 4.5) .. (c5);
			\draw (d) -- (d1) -- (d2) -- (d3) -- (d4) -- (d5);
            \draw (d) .. controls (2.5,0.5) and (2.5, 4.5) .. (d5);
			\draw (e) -- (e1) -- (e2) -- (e3) -- (e4) -- (e5);
             \draw (e) .. controls (3.5,0.5) and (3.5, 4.5) .. (e5);
			\draw (f) -- (f1) -- (f2) -- (f3) -- (f4) -- (f5);
             \draw (f) .. controls (4.5,0.5) and (4.5, 4.5) .. (f5);

            	\node[circle, draw, fill = black, scale=0.75] (a') at (7,0) {};
                \node at (a'.west) [left=8pt] {\footnotesize $u_1$};
            \node at (a'.south) [below=8pt] {\footnotesize $v_1$};
			\node[circle, draw, fill = black, scale=0.75] (b') at (8,0) {};
            \node at (b'.south) [below=8pt] {\footnotesize $v_2$};
			\node[circle, draw, fill = black, scale=0.75] (c') at (9,0) {};
            \node at (c'.south) [below=8pt] {\footnotesize $v_3$};
			\node[circle, draw, scale=0.75] (d') at (10,0) {}; 
            \node at (d'.south) [below=8pt] {\footnotesize $v_4$};
			\node[circle, draw, fill = black, scale=0.75] (e') at (11,0) {}; 
            \node at (e'.south) [below=8pt] {\footnotesize $v_5$};
			\node[circle, draw, fill = black, scale=0.75] (f') at (12,0) {}; 
            \node at (f'.south) [below=8pt] {\footnotesize $v_6$};
			\node[circle, draw, fill = black, scale=0.75] (g') at (13,0) {};
            \node at (g'.south) [below=8pt] {\footnotesize $v_7$};
            \node[circle, draw, fill = black, scale=0.75] (h') at (14,0) {};
            \node at (h'.south) [below=8pt] {\footnotesize $v_8$};
			\draw (a') -- (b') -- (c') -- (d') -- (e') -- (f') -- (g') -- (h');
           \draw (a') .. controls (7.5,-0.5) and (13.5, -0.5) .. (h');
			\node[circle, draw, scale=0.75] (a1') at (7,1) {};
            \node at (a1'.west) [left=8pt] {\footnotesize $u_2$};
			\node[circle, draw, scale=0.75] (b1') at (8,1) {};
			\node[circle, draw, fill = black, scale=0.75] (c1') at (9,1) {};
			\node[circle, draw, scale=0.75] (d1') at (10,1) {}; 
			\node[circle, draw, fill = black, scale=0.75] (e1') at (11,1) {};  
			\node[circle, draw, scale=0.75] (f1') at (12,1) {};  
			\node[circle, draw, scale=0.75] (g1') at (13,1) {};
            \node[circle, draw, scale=0.75] (h1') at (14,1) {};
			\draw (a1') -- (b1') -- (c1') -- (d1') -- (e1') -- (f1') -- (g1') -- (h1'); 
            \draw (a1') .. controls (7.5,0.5) and (13.5, 0.5) .. (h1');
			\node[circle, draw, fill = black, scale=0.75] (a2') at (7,2) {};
            \node at (a2'.west) [left=8pt] {\footnotesize $u_3$};
			\node[circle, draw, scale=0.75] (b2') at (8,2) {};
			\node[circle, draw, fill = black, scale=0.75] (c2') at (9,2) {};
			\node[circle, draw, scale=0.75] (d2') at (10,2) {}; 
			\node[circle, draw, fill = black, scale=0.75] (e2') at (11,2) {};  
			\node[circle, draw, scale=0.75] (f2') at (12,2) {};  
			\node[circle, draw, fill = black, scale=0.75] (g2') at (13,2) {};
            \node[circle, draw, fill = black, scale=0.75] (h2') at (14,2) {};
			\draw (a2') -- (b2') -- (c2') -- (d2') -- (e2') -- (f2') -- (g2') -- (h2');
            \draw (a2') .. controls (7.5,1.5) and (13.5, 1.5) .. (h2');
			\node[circle, draw, fill = black, scale=0.75] (a3') at (7,3) {};
            \node at (a3'.west) [left=8pt] {\footnotesize $u_4$};
			\node[circle, draw, scale=0.75] (b3') at (8,3) {};
			\node[circle, draw, scale=0.75] (c3') at (9,3) {};
			\node[circle, draw, scale=0.75] (d3') at (10,3) {$x$}; 
			\node[circle, draw, scale=0.75] (e3') at (11,3) {};  
			\node[circle, draw, scale=0.75] (f3') at (12,3) {};  
			\node[circle, draw, scale=0.75] (g3') at (13,3) {};
            \node[circle, draw, scale=0.75] (h3') at (14,3) {};
			\draw (a3') -- (b3') -- (c3') -- (d3') -- (e3') -- (f3') -- (g3') -- (h3');
            \draw (a3') .. controls (7.5,2.5) and (13.5, 2.5) .. (h3');
			\node[circle, draw, fill = black, scale=0.75] (a4') at (7,4) {};
            \node at (a4'.west) [left=8pt] {\footnotesize $u_5$};
			\node[circle, draw, scale=0.75] (b4') at (8,4) {};
			\node[circle, draw, fill = black, scale=0.75] (c4') at (9,4) {};
			\node[circle, draw, scale=0.75] (d4') at (10,4) {}; 
			\node[circle, draw, fill = black, scale=0.75] (e4') at (11,4) {};  
			\node[circle, draw, scale=0.75] (f4') at (12,4) {};  
			\node[circle, draw, fill = black, scale=0.75] (g4') at (13,4) {};
            \node[circle, draw, scale=0.75] (h4') at (14,4) {};
			\draw (a4') -- (b4') -- (c4') -- (d4') -- (e4') -- (f4') -- (g4') -- (h4');
            \draw (a4') .. controls (7.5,3.5) and (13.5, 3.5) .. (h4');
			\node[circle, draw, fill = black, scale=0.75] (a5') at (7,5) {};
            \node at (a5'.west) [left=8pt] {\footnotesize $u_6$};
			\node[circle, draw, scale=0.75] (b5') at (8,5) {};
			\node[circle, draw, fill = black, scale=0.75] (c5') at (9,5) {};
			\node[circle, draw, scale=0.75] (d5') at (10,5) {}; 
			\node[circle, draw, fill = black, scale=0.75] (e5') at (11,5) {};  
			\node[circle, draw, scale=0.75] (f5') at (12,5) {};
			\node[circle, draw, fill = black, scale=0.75] (g5') at (13,5) {}; 
            \node[circle, draw, fill = black, scale=0.75] (h5') at (14,5) {}; 
			\draw (a5') -- (b5') -- (c5') -- (d5') -- (e5') -- (f5') -- (g5') -- (h5');
            \draw (a5') .. controls (7.5,4.5) and (13.5, 4.5) .. (h5');
			\node[circle, draw, scale=0.75] (a6') at (7,6) {};
            \node at (a6'.west) [left=8pt] {\footnotesize $u_7$};
			\node[circle, draw, scale=0.75] (b6') at (8,6) {};
			\node[circle, draw, fill = black, scale=0.75] (c6') at (9,6) {};
			\node[circle, draw, scale=0.75] (d6') at (10,6) {}; 
			\node[circle, draw, fill = black, scale=0.75] (e6') at (11,6) {};  
			\node[circle, draw, scale=0.75] (f6') at (12,6) {};
			\node[circle, draw, scale=0.75] (g6') at (13,6) {}; 
            \node[circle, draw, scale=0.75] (h6') at (14,6) {}; 
			\draw (a6') -- (b6') -- (c6') -- (d6') -- (e6') -- (f6') -- (g6') -- (h6');
            \draw (a6') .. controls (7.5,5.5) and (13.5, 5.5) .. (h6');
\node[circle, draw, fill = black, scale=0.75] (a7') at (7,7) {};
\node at (a7'.west) [left=8pt] {\footnotesize $u_8$};
			\node[circle, draw, fill = black, scale=0.75] (b7') at (8,7) {};
			\node[circle, draw, fill = black, scale=0.75] (c7') at (9,7) {};
			\node[circle, draw, scale=0.75] (d7') at (10,7) {}; 
			\node[circle, draw, fill = black, scale=0.75] (e7') at (11,7) {};  
			\node[circle, draw, fill = black, scale=0.75] (f7') at (12,7) {};
			\node[circle, draw, fill = black, scale=0.75] (g7') at (13,7) {}; 
            \node[circle, draw, fill = black, scale=0.75] (h7') at (14,7) {}; 
			\draw (a7') -- (b7') -- (c7') -- (d7') -- (e7') -- (f7') -- (g7') -- (h7');
\draw (a7') .. controls (7.5,6.5) and (13.5, 6.5) .. (h7');
            
			\draw (a') -- (a1') -- (a2') -- (a3') -- (a4') -- (a5') -- (a6') -- (a7');
            \draw (a') .. controls (6.5,0.5) and (6.5, 6.5) .. (a7');
			\draw (b') -- (b1') -- (b2') -- (b3') -- (b4') -- (b5') -- (b6') -- (b7');
            \draw (b') .. controls (7.5,0.5) and (7.5, 6.5) .. (b7');
			\draw (c') -- (c1') -- (c2') -- (c3') -- (c4') -- (c5') -- (c6') -- (c7');
            \draw (c') .. controls (8.5,0.5) and (8.5, 6.5) .. (c7');
			\draw (d') -- (d1') -- (d2') -- (d3') -- (d4') -- (d5') -- (d6') -- (d7');
            \draw (d') .. controls (9.5,0.5) and (9.5, 6.5) .. (d7');
			\draw (e') -- (e1') -- (e2') -- (e3') -- (e4') -- (e5') -- (e6') -- (e7');
            \draw (e') .. controls (10.5,0.5) and (10.5, 6.5) .. (e7');
			\draw (f') -- (f1') -- (f2') -- (f3') -- (f4') -- (f5') -- (f6') -- (f7');
            \draw (f') .. controls (11.5,0.5) and (11.5, 6.5) .. (f7');
			\draw (g') -- (g1') -- (g2') -- (g3') -- (g4') -- (g5') -- (g6') -- (g7');
            \draw (g') .. controls (12.5,0.5) and (12.5, 6.5) .. (g7');
            \draw (h') -- (h1') -- (h2') -- (h3') -- (h4') -- (h5') -- (h6') -- (h7');
            \draw (h') .. controls (13.5,0.5) and (13.5, 6.5) .. (h7');
	\end{tikzpicture}
		\caption{vv-sets of $C_6 \cp C_6$ and $C_8 \cp C_8$.}
		\label{fig:eventorus}
	\end{center}
\end{figure}

Next assume that $n$ is even and $x = (u_{\frac{n}{2}},v_{\frac{n}{2}})$. For each $(u_k,v_l) \in Q_t$, we know that the set $N((u_k,v_l)) \cap I(x,(u_k, v_l))$ has cardinality at most four and has a vertex which is not in $S$. Particularly for $(u_k,v_l) \in Q_1$, this vertex is either $(u_k,v_{l+1})$ or $(u_{k+1},v_l)$. As in the previous two proofs, partition each $Q_t$ into diagonals $D_{pq}$, where $p, q \in Q_t$ are its end vertices. Then we can have at most $\sum_{pq} \lceil |D_{pq}| / 2 \rceil$ vertices in $S \cap \bigcup_{t=1}^{4} Q_t$. Thus, for $n \equiv 0 \bmod 4$ we get  
\begin{align*}
|S \cup \displaystyle \bigcup_{t=1}^{4} Q_t| & = |S \cap Q_1| + |S \cap Q_2| + |S \cap Q_3| + |S \cap Q_4| \\
& \le\frac{(n-2)n}{8} + \frac{n^2}{8} + \frac{n(n+2)}{8} + \frac{n^2}{8} \\
& = \frac{n^2}{2}\,,
\end{align*}    
and for $n \equiv 2 \bmod 4$ we get
\begin{align*}
|S \cup \displaystyle \bigcup_{t=1}^{4} Q_t| & = |S \cap Q_1| + |S \cap Q_2| + |S \cap Q_3| + |S \cap Q_4| \\
& \le\frac{(n-2)n}{8} + \frac{n^2 - 4}{8} + \frac{n(n+2)}{8} + \frac{n^2 - 4}{8} \\
& = \frac{n^2 -2}{2}\,.
\end{align*}   
Again, by the choice of $S$, every vertex must be $S$-visible from $x$ and hence, $S \cap (X \cup Y) \subseteq \{(u_{\frac{n}{2}},v_1), (u_{\frac{n}{2}},v_{n-1}), (u_{\frac{n}{2}},v_n), (u_1,v_{\frac{n}{2}}), (u_{n-1},v_{\frac{n}{2}}), (u_n,v_{\frac{n}{2}})\}$. Also, 
\begin{align*}
\text{in $Q_1$}: &\quad |D_{(u_\frac{n-2}{2},v_2)(u_2,v_\frac{n-2}{2})}| =  \frac{n-4}{2}, \\
\text{in $Q_2$}: &\quad |D_{(u_\frac{n-2}{2},v_{n-2})(u_2,v_\frac{n+2}{2})}| =  \frac{n-4}{2} \text{ and } |D_{(u_\frac{n-2}{2},v_{n-1})(u_1,v_\frac{n+2}{2})}| = \frac{n-2}{2}, \\
\text{in $Q_3$}: &\quad |D_{(u_{n-2},v_\frac{n+2}{2})(u_\frac{n+2}{2},v_{n-2})}| =  \frac{n-4}{2} \text{ and } |D_{(u_{n-1},v_\frac{n+2}{2})(u_\frac{n+2}{2},v_{n-1})}| = \frac{n-2}{2},\\
\text{in $Q_4$}: &\quad |D_{(u_{n-2},v_\frac{n-2}{2})(u_\frac{n+2}{2},v_{3})}| =  \frac{n-4}{2} \text{ and } |D_{(u_{n-1},v_\frac{n-2}{2})(u_\frac{n+2}{2},v_{2})}| = \frac{n-2}{2}.
\end{align*}
Then, since we have assumed that $S$ has the maximum number of vertices from  $\bigcup_{t=1}^{4} Q_t$, we get 
\begin{flalign*}
  |S \cap (X \cup Y)|= \begin{cases}
      1; & n \equiv 0 \bmod 4,\\
      2; & n \equiv 2 \bmod 4.
  \end{cases} \label{}
\end{flalign*}
Thus for $n \equiv 0 \bmod 4$, 
$$|S| \leq \frac{n^2 +2}{2}.$$

To prove the required lower bound, let $R$ be the set of $\lceil |D_{pq}|/2 \rceil$ alternate vertices from each $D_{pq}$; see Figures~\ref{fig:oddtorus} and~\ref{fig:eventorus} which illustrate the construction on $P_5 \cp C_5$, $P_6 \cp C_6$, $P_7 \cp C_7$, and $P_8 \cp P_8$. Then $R$ is an $x$-visibility set of cardinality $\frac{n^2 -1}{2}$ for $n \equiv 1 \bmod 4$ and $R\cup \{(u_{\frac{n+1}{2}},v_1), (u_{\frac{n+1}{2}},v_n)\}$ is an $x$-visibility set of cardinality $\frac{n^2 +3}{2}$ for $n \equiv 3 \bmod 4$. Also, $R\cup \{(u_{\frac{n+1}{2}},v_1)\}$ is an $x$-visibility set of cardinality $\frac{n^2 +2}{2}$ for $n \equiv 0 \bmod 4$ and $R\cup \{(u_{\frac{n+1}{2}},v_n), (u_n,v_{\frac{n+1}{2}})\}$ is an $x$-visibility set of cardinality $\frac{n^2 +2}{2}$ for $n \equiv 2 \bmod 4$. 
\qed
\end{proof}

The results of Theorems~\ref{vv-sqauregrid}, \ref{vv-sqaurecylinder} and \ref{vv-sqauretorus} are summarized in Table~\ref{table:all-values}.

\begin{table}[ht!]
\begin{center}
\begin{tabular}{ ||c||c|c|c|c|| } 
\hline
\hline
 & $n \equiv 1 \bmod 4$ & $n \equiv 3 \bmod 4$ & $n \equiv 0 \bmod 4$ & $n \equiv 2 \bmod 4$ 
\phantom{$\begin{bmatrix} X \\ Y \end{bmatrix}$} \hspace*{-10mm}\\
\hline
\hline
$P_n \cp P_n$ & $\frac{n^2+n-2}{2}$ & $\frac{n^2+n-2}{2}$ & $\frac{n^2+n-2}{2}$ & $\frac{n^2+n-2}{2}$ 
\phantom{$\begin{bmatrix} X \\ Y \end{bmatrix}$} \hspace*{-10mm}\\ 
\hline
$P_n \cp C_n$ & $\frac{n^2+3}{2}$ & $\frac{n^2+n-2}{2}$ & $\frac{2n^2+n}{4}$ & $\frac{2n^2+n-2}{4}$  
\phantom{$\begin{bmatrix} X \\ Y \end{bmatrix}$} \hspace*{-10mm} \\ 
\hline
$C_n \cp C_n$ & $\frac{n^2-1}{2}$ & $\frac{n^2+3}{2}$ & $\frac{n^2+2}{2}$ & $\frac{n^2+2}{2}$ 
\phantom{$\begin{bmatrix} X \\ Y \end{bmatrix}$} \hspace*{-10mm} \\ 
\hline
\hline
\end{tabular}
\end{center}
\caption{The vertex visibility number of square grids, prisms, and toruses}
\label{table:all-values}
\end{table}

\section{Concluding remarks}

We conclude the paper with some problems that deserve attention in the future. 

Related to the sharp example of Proposition~\ref{prop:Cart} we pose: 

\begin{problem}
Determine the vertex visibility number of the Cartesian product of an arbitrary product of finitely many complete graphs. 
\end{problem}

Having studied the vertex visibility problem in the context of the Cartesian product, it is also worthwhile to explore it in other graph products. In particular, mutual-visibility has already been studied on strong products~\cite{cicerone-2024a}, hence we pose: 

\begin{problem}
Study the vertex visibility in the strong product of two graphs. 
\end{problem}

The variety of mutual-visibility problems and the variety of the general position problems have been studied in \cite{{Roy-2026}} on Sierpi\'nski graphs $S_3^n$, $n \geq 3$. Therefore, investigating the vertex visibility problem in these graphs may also be of interest.

\begin{problem}
Determine the vertex visibility number of the Sierpi\'nski graphs $S_3^n$, $n \geq 3$. 
\end{problem}

\section*{Declaration of interests}
 
The authors declare that they have no conflict of interest. 

\section*{Data availability}
 
Our manuscript has no associated data.

\section*{Acknowledgments}
	
The authors thank Ullas Chandran S V for the initial idea for this research, which greatly influenced the direction of this project. We also thank Manoj Changat for some initial discussions. Dhanya Roy thank Cochin University of Science and Technology for providing financial support under University JRF Scheme.  Sandi Klav\v{z}ar was supported by the Slovenian Research and Innovation Agency (ARIS) under the grants P1-0297, N1-0285, N1-0355, N1-0431. Gabriele {Di Stefano} was supported by the Italian National Group for Scientific Computation (GNCS-INdAM) and by the academic project ``Monet'' of the University of L'Aquila.

\end{document}